\documentclass[reqno]{amsart}
\usepackage[foot]{amsaddr}
\usepackage{graphicx}
\graphicspath{ {./figures/} }
\usepackage[margin=3cm]{geometry}
\usepackage{amsmath, amssymb,amsthm}
\usepackage[scr=rsfs]{mathalpha}
\usepackage[shortlabels]{enumitem}
\usepackage{mlmodern}
\DeclareSymbolFont{largesymbols}{OMX}{cmex}{m}{n} 
\usepackage{mathtools} 
\usepackage{stmaryrd} 
\usepackage{tikz}
\usetikzlibrary{decorations.pathreplacing,patterns,math,external,decorations.pathmorphing,arrows.meta}
\usepackage{upref} 
\usepackage[colorlinks]{hyperref}
\hypersetup{citecolor=blue,filecolor=blue,linkcolor=blue,urlcolor=navyblue}
\definecolor{navyblue}{rgb}{0.0, 0.0, 0.5}

\newtheorem{thm}{Theorem}[section]
\newtheorem{prop}[thm]{Proposition}
\newtheorem{lem}[thm]{Lemma}

\theoremstyle{definition}

\newtheorem{rmk}{Remark}[section]

\newtheorem*{claim*}{Claim}
\newtheorem*{fact*}{Useful fact}

\newcommand{\Z}{\mathbb{Z}}
\newcommand{\R}{\mathbb{R}}
\newcommand{\C}{\mathbb{C}}

\newcommand{\Sb}{\mathbb{S}}

\renewcommand{\Re}{\operatorname{\mathrm{Re}}}
\renewcommand{\mod}{\,\operatorname{mod}\,}

\newcommand{\oneb}{\mathbf{1}}
\renewcommand{\P}{\mathbf{P}}

\newcommand{\supp}{\operatorname{supp}}

\newcommand{\sgn}{\operatorname{sgn}}

\DeclareMathAlphabet{\mathcalboondox}{U}{BOONDOX-calo}{m}{n}

\newcommand{\psiout}{\psi_{\mathrm{out}}}

\newcommand{\SU}{\mathrm{SU}}

\newcommand\numberthis{\stepcounter{equation}\tag{\theequation}}
\numberwithin{equation}{section}

\newcommand\reds[1]{{\color{red}#1}}

\begin{document}

\title[]{Fast GHZ encoding with $1-o(1)$ fidelity}
\author{%
Laura Shou$^{1}$,
T. C. Mooney$^{1,2}$,
Twesh Upadhyaya$^{1,2}$,
Alexey V. Gorshkov$^{1,2}$
}

\address{\normalfont$^1$Joint Quantum Institute, Department of Physics, NIST/University of Maryland, College Park, MD 20742, USA}
\address{\normalfont$^2$Joint Center for Quantum Information and Computer Science,
NIST/University of Maryland, College Park, MD, 20742, USA}

\begin{abstract}
We prove that $(1-o(1))$-fidelity $N$-qubit Greenberger--Horne--Zeilinger (GHZ) encoding can be performed in time $O(\log N/N)$ using all-to-all 2-local Hamiltonians with bounded 2-qubit interaction terms. This saturates the theoretical lower bound $\Omega(\log N/N)$, and by rescaling, also saturates the lower bound on signaling time for Hamiltonians with power-law interaction strengths $1/r^\gamma$, $\gamma<d$ in dimension $d$. The protocol uses spin-squeezing dynamics combined with quantum signal processing.
\end{abstract}
\maketitle

\section{Introduction}

The $N$-qubit Greenberger--Horne--Zeilinger (GHZ) state $\frac{1}{\sqrt{2}}(|0^N\rangle+|1^N\rangle)$ \cite{greenberger1989going} and its relatives $\alpha|0^N\rangle+\beta|1^N\rangle$, $\alpha,\beta\ne0$, are entangled states with many applications in quantum information and computing, including to quantum metrology \cite{bollinger1996optimal}, quantum cryptography \cite{hillery1999quantum,murta2020quantum}, and quantum error correction \cite{gottesman2010introduction}.
Quantum information carried by an unknown single-qubit state $\alpha|0\rangle+\beta|1\rangle$ can be encoded into the entangled GHZ-like state $\alpha|0^N\rangle+\beta|1^N\rangle$ under unitary evolution, after which the initial quantum information is distributed throughout the system.
In this paper, we consider how fast one can perform GHZ encoding using all-to-all interacting Hamiltonians. By a simple scaling argument, this also answers the question for Hamiltonians with long-range power-law interaction strengths $1/r^\gamma$, $\gamma<d$ in dimension $d$. 
We consider all-to-all 2-local Hamiltonians with arbitrarily strong single-qubit terms and bounded 2-qubit interaction terms,
\begin{align}\label{eqn:hamiltonian}
H(t)=\sum_ih_i^a(t)X_i^a+\sum_{i<j}J_{ij}^{ab}(t)X_i^aX_j^b,\quad |J_{ij}^{ab}(t)|\le1,
\end{align}
for identity or Pauli operators $X_i^a,X_j^b$, and implicit sums over $a,b$.
For a system of qubits indexed by $\Lambda=\{0,1,\ldots,N\}$, we are interested in the GHZ encoding problem: Writing $\Lambda=\{q\}\cup  A$, we want to use a unitary operation $U$ corresponding to time evolution of the Hamiltonian \eqref{eqn:hamiltonian} to encode the single qubit state $\alpha|0\rangle_q+\beta|1\rangle_q$ into the GHZ subspace,
\begin{align}\label{eqn:encoding}
(\alpha|0\rangle+\beta|1\rangle)_q\otimes|\mathbf 0\rangle_{A}\overset{U}{\longmapsto}\alpha|\mathbf 0\rangle_\Lambda+\beta|\mathbf1\rangle_\Lambda,
\end{align}
up to a small error in the final state. 

As GHZ encoding spreads the initial quantum information throughout the system, and can also be undone to perform state transfer \cite{tran2021optimal}, the GHZ encoding time $T$ is restricted by Lieb--Robinson bounds \cite{lieb1972finite,hastings2006spectral,nachtergaele2010lieb,chen2023speed}. 
For short-range systems, the Lieb--Robinson bound gives a linear light cone for information propagation, and thus an encoding time lower bound of $T\ge\Omega(N^{1/d})$ in dimension $d$.
For long-range systems with power-law decay in the distance $r$, information may travel much faster \cite{tran2021optimal,chen2023speed,yin2025fast,Upadhyaya2026,Upadhyaya2026a}.
Here we consider the longest-range regimes: For all-to-all Hamiltonians, \cite{yin2025fast,chen2023speed} give the lower bound $T\ge\Omega(\frac{\log N}{N})$.
For the long-range power-law interactions decaying as $1/r^\gamma$ in the distance $r$, with $\gamma<d$ in dimension $d$,  \cite{guo2020signaling} gives the lower bound $T\ge\Omega(\frac{\log N}{N^{1-\gamma/d}})$. 
In this paper, we give a $(1-o(1))$-fidelity GHZ encoding protocol saturating these bounds. 

The first fast all-to-all protocol is \cite{yin2025fast}, which produces a numerically high-fidelity GHZ encoding in time $O((\log N)^2/N)$, nearly saturating the theoretical lower bounds. However, there is no rigorous guarantee for the fidelity of the protocol; as noted in \cite{yin2025fast}, the infidelity appears numerically non-vanishing, and may even increase slightly at large $N$ for fixed protocol parameters.
Therefore, to fully resolve the speed limit question in this regime, it remains to find a rigorously $1-\varepsilon$ or $1-o(1)$ fidelity fast protocol, as well as to obtain the optimal time scaling $O(\frac{\log N}{N})$.

We use a combination of fast squeezing operations and quantum signal processing to address both of these problems and prove 
\begin{thm}\label{thm:encoding}
GHZ encoding with $(1-o(1))$-fidelity can be performed using a Hamiltonian evolution of the form in \eqref{eqn:hamiltonian} in time $O(\log N/N)$.
\end{thm}
We give the explicit construction of $U$ and $H(t)$ in Section~\ref{subsec:protocol} and Theorem~\ref{thm:protocol}.
The time $O(\log N/N)$ saturates the lower bound of time $\Omega(\log N/N)$ from \cite{yin2025fast} for producing a $(1/2+\delta)$-fidelity GHZ encoding. By rescaling the $J_{ij}^{ab}$, this protocol also saturates the signaling lower bound $\Omega(\log N/(N^{1-\gamma/d}))$ from \cite{guo2020signaling} for long-range power-law interaction strengths $1/r^\gamma$, $\gamma<d$, in dimension $d$.
The protocol should also imply saturating GHZ encoding bounds with $k$-local Hamiltonians \cite{guo2020signaling,yin2025fastq}.

We briefly describe some of the ideas of the protocol here, with a more detailed description in Section~\ref{subsec:protocol} and Figure~\ref{fig:protocol}.
The protocol uses two qubit registers $A$ and $B$, so the system of qubits is $\Lambda=\{q\}\cup A\cup B$. On $A$, we start with two-axis twisting (TAT) and controlled $X$ rotation as in \cite{yin2025fast}. The idea there is to squeeze the packet, then use the data qubit $\alpha|0\rangle_q+\beta|1\rangle_q$ to perform a quick controlled $X$ rotation which separates the narrow packets, and then perform reverse TAT squeezing to send the packets far apart. Instead of proceeding with a refocusing procedure to physically produce an approximate GHZ encoding, we next use a quantum signal processing (QSP) method. The QSP method reads analog values from $A$, and writes $0$s or $1$s as appropriate to register $B$. After resetting the qubits in $A$, the $0$s or $1$s in $B$ can be used to copy the values back to $A$ quickly, which will produce a provably high fidelity GHZ encoding. Using the QSP method allows for taking higher TAT squeezing in the first step of the protocol, which allows for a shorter controlled $X$ rotation, which removes a factor of $\log N$ from the protocol time. The modified TAT-rotate-TAT method combined with the QSP method also allows for rigorously tracking the fidelity throughout the protocol.

\subsection{Spin-squeezing protocol of \texorpdfstring{\cite{yin2025fast}}{[Yin25]}}
Before describing the protocol in Theorem~\ref{thm:encoding} in more detail, we start by briefly reviewing parts of the protocol from \cite{yin2025fast}. The protocol applies a sequence of fast spin-squeezing dynamics to obtain a numerically accurate GHZ encoding in time $O((\log N)^2/N)$. The data qubit $\alpha|0\rangle+\beta|1\rangle$ is used as a control, so that applying a two-axis-twisting (TAT) procedure plus controlled rotation on the remaining qubits can produce two separate squeezed packets. Another TAT squeezing can then rapidly send these packets in opposite directions to opposite $Y$ poles of the Dicke sphere.
A refocusing procedure consisting of single-qubit rotations, one-axis-twisting (OAT), and another TAT unitary, then transforms the two squeezed packets into an approximate GHZ encoding $\alpha|0^N\rangle+\beta|1^N\rangle$.

We focus on the first part of the procedure, which is a 3-step TAT-rotate-TAT procedure which produces two packets at opposite $Y$ poles. 
Let $A$ index all the qubits except for the data qubit $q$. Define total spin components $X:=\sum_{i\in  A}X_i$, $Y:=\sum_{i\in A}Y_i$, and $Z:=\sum_{i\in A}Z_i$.
The 3-step procedure sends the initial state $\alpha|0\rangle_q|0^{N}\rangle_A+\beta|1\rangle_q|0^N\rangle_A$ to 
\begin{align}\label{eqn:tat0}
\alpha|0\rangle_q|\chi_+\rangle_A+\beta|1\rangle_q|\chi_-\rangle_A,
\end{align}
where $|\chi_\pm\rangle_A$ are squeezed states located near opposite $Y$ poles of the Dicke sphere.
The 3 steps are: TAT squeezing, controlled $X$ rotation, then reverse TAT squeezing, in total expressed by the unitary operation
\begin{align}\label{eqn:tatv}
V=e^{i\tau_2\frac{\log N}{N}2K}e^{i\frac{\phi}{2}Z_qX}e^{-i\tau_1\frac{\log N}{N}2K},
\end{align}
for $K:=\frac12(XY+YX)$ the TAT generator, and for $\tau_1,\tau_2=O(1)$ and $\phi=\theta(\log N)^2/N$, with $\tau_1,\tau_2,\theta$ to be optimized numerically.
The action of each of the three steps on the Dicke sphere is shown visually in the top half of Fig. 2(c) of \cite{yin2025fast}. For comparison with the new protocol, we also depict these steps in Figure~\ref{fig:bloch3}.
\begin{enumerate}[1.]
\item The first TAT unitary, $e^{-i\tau_1\frac{\log N}{N}2K}$, squeezes the state $|0\rangle_A$ to a state $|\eta\rangle_A$ which is narrow in $Y$ and long in $X$. The extreme squeezing time is $\tau_\mathrm{min}\approx1/8$, at which point $\Delta Y\sim1$; in continuum or normalized Dicke sphere coordinates $y\in[-1,1]$, this corresponds to $\Delta y\sim1/N$. However, taking extreme squeezing causes difficulties with the later refocusing procedure, so $\tau_1$ is taken below $\tau_\mathrm{min}$. 

\item The next step, $e^{i\frac{\phi}{2}Z_qX}$, is a controlled rotation in $X$. This rotates the two packets $\alpha|0\rangle_q|\eta\rangle_A$ and $\beta|1\rangle_q|\eta\rangle_A$ in opposite directions, to say $\alpha|0\rangle_q|\eta_+\rangle_A$ and $\beta|1\rangle_q|\eta_-\rangle_A$, respectively. Because $|\eta\rangle_A$ is narrow in $Y$, a small rotation in $X$ is enough to separate the state into two separate packets.
However, because $\tau_1$ is taken below $\tau_\mathrm{min}$, the packet is wider, and the rotation time $\phi$ has to be taken longer than the $O(\log N/N)$ that one would ideally like to take.

\item The final step, $e^{i\tau_2\frac{\log N}{N}(2K)}$, is reverse TAT squeezing. It stretches in the $Y$ direction. Since the two packets $|\eta_\pm\rangle_A$ are not centered in $Y$, this squeezing moves their centers apart, to near the $Y$ poles. The resulting state is \eqref{eqn:tat0}.
\end{enumerate}

\begin{figure}[htb]
\includegraphics[width=\textwidth]{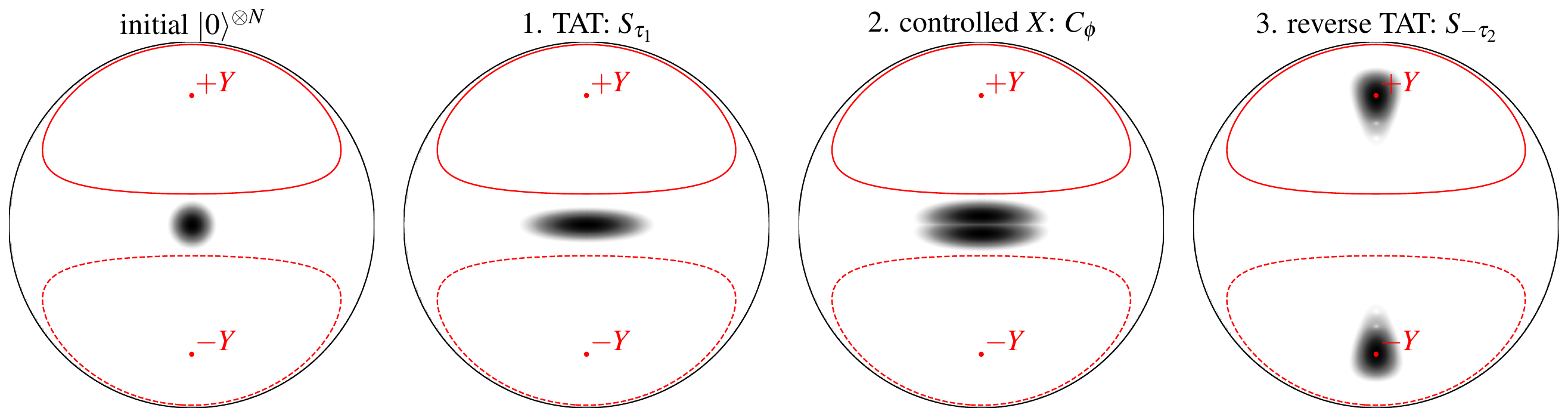}
\caption{Illustration of the first three steps (TAT, controlled $X$ rotation, reverse TAT) of the protocol in \cite{yin2025fast} on the Dicke sphere with $N=1000$. The Husimi distribution of the current state of each packet is shown overlayed in log scale, using a Lambert azimuthal projection of the sphere with the center being the $+Z$ pole.
Note that due to the log scale, the packets in the third panel after the controlled $X$ step are actually much more separated than they may appear.
While not relevant for the protocol in \cite{yin2025fast}, the red contours show the boundaries $y=\pm1/3$, to compare with Figure~\ref{fig:bloch}.
}\label{fig:bloch3}
\end{figure}

\begin{figure}[htb]
\includegraphics[width=\textwidth]{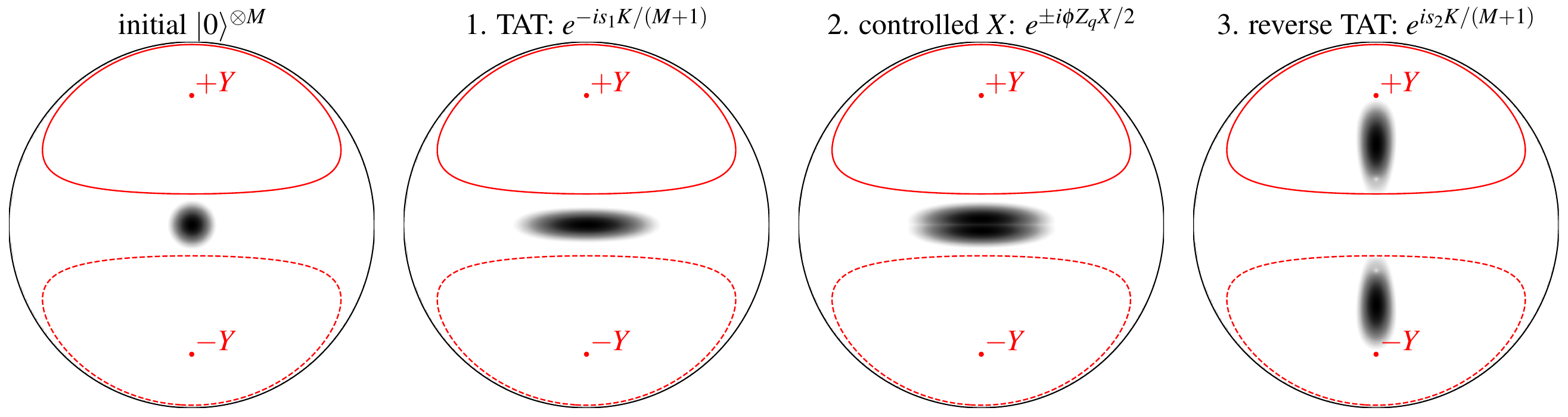}
\caption{
Illustration of the TAT packet separation in the new protocol, which consists of the same steps as the start of the protocol from \cite{yin2025fast} shown in Figure~\ref{fig:bloch3}, but with different squeezing and rotation parameters. The parameters used for the three steps are those in the table in Section~\ref{subsec:outline} with $a=9/20$, $\kappa=9$, and $M=1000$.
The red contours show the boundaries $y=\pm1/3$. In Step 1, we squeeze the state further than in Figure~\ref{fig:bloch3}, while in Step 2, we do a shorter controlled $X$ rotation. In Step 3, we squeeze a smaller amount to send the ``core'' of the packet past the $y=\pm1/3$ boundaries, but not all the way to the $\pm Y$ poles. Stopping before reaching the $\pm Y$ poles lets us handle the exact solution to the continuum limit TAT evolution equations, which allows for obtaining rigorous estimates.
Note that the packet densities (Husimi distributions) are shown in log scale, and so the packets in the controlled $X$ step are actually much more separated than they may appear.
}\label{fig:bloch}
\end{figure}

\subsection{New protocol with QSP} \label{subsec:protocol}

We use two qubit registers $A$ and $B$. On $A$, we perform a similar 3-step TAT-rotate-TAT procedure as described above. 
Recall the second half of the protocol of \cite{yin2025fast} is a refocusing procedure to transform the two squeezed packets in \eqref{eqn:tat0} into an approximate GHZ encoding.
We replace this refocusing procedure with a quantum signal processing (QSP) read/write method to register $B$, which will allow us to take close to extreme squeezing $\tau_1\approx\tau_{\mathrm{min}}$ in the first TAT squeezing on $A$, since there will be no need to control or refocus the resulting packets. Heuristically, on $A$, the packet after the first TAT squeezing will have width $\Delta y\sim O(1/N)$, so taking a shorter controlled $X$ rotation of duration $\phi=\Theta(\log N/N)$, which moves the packets a distance $\Theta(\log N/N)$ apart, is enough to separate into two separate packets. An illustration is shown in Figure~\ref{fig:bloch}.
In order to rigorously bound the fidelity of the final state, we will use a finite-difference relation and continuum evolution to keep track of quantitative error estimates.

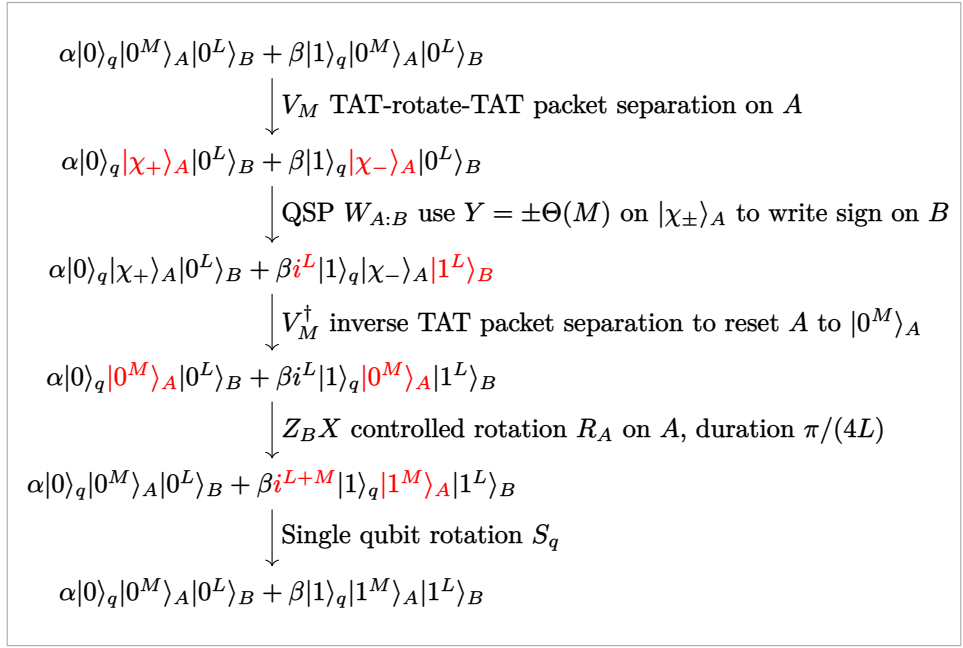
\begin{figure}[htb]
\begin{tikzpicture}
\def\dl{.75cm}
\def\step{1.9*.75cm}
\node [above] at (0,0) {$\alpha|0\rangle_q{|0^M\rangle_A}|0^L\rangle_B+\beta|1\rangle_q{|0^M\rangle_A}|0^L\rangle_B$};
\draw[->] (0,0)--++(0,-\dl);
\node[right] at (0,-\dl/2) {$V_M$ TAT-rotate-TAT packet separation on $A$};
\begin{scope}[yshift={-\step}]
\node [above] at (0,0) {$\alpha|0\rangle_q\reds{|\chi_+\rangle_A}|0^L\rangle_B+\beta|1\rangle_q\reds{|\chi_-\rangle_A}|0^L\rangle_B$};
\draw[->] (0,0)--++(0,-\dl);
\node[right] at (0,-\dl/2) {QSP $W_{A:B}$ use $Y=\pm\Theta(M)$ on $|\chi_\pm\rangle_A$ to write sign on $B$};
\end{scope}
\begin{scope}[yshift={-2*\step}]
\node [above] at (0,0) {$\alpha|0\rangle_q|\chi_+\rangle_A|0^L\rangle_B+\beta \reds{i^L}|1\rangle_q|\chi_-\rangle_A\reds{|1^L\rangle_B}$};
\draw[->] (0,0)--++(0,-\dl);
\node[right] at (0,-\dl/2) {$V_M^\dagger$ inverse TAT packet separation to reset $A$ to $|0^M\rangle_A$};
\end{scope}
\begin{scope}[yshift={-3*\step}]
\node [above] at (0,0) {$\alpha|0\rangle_q\reds{|0^M\rangle_A}|0^L\rangle_B+\beta i^L|1\rangle_q\reds{|0^M\rangle_A}|1^L\rangle_B$};
\draw[->] (0,0)--++(0,-\dl);
\node[right] at (0,-\dl/2) {$Z_BX$ controlled rotation $R_A$ on $A$, duration $\pi/(4L)$};
\end{scope}
\begin{scope}[yshift={-4*\step}]
\node [above] at (0,0) {$\alpha|0\rangle_q|0^M\rangle_A|0^L\rangle_B+\beta \reds{i^{L+M}}|1\rangle_q\reds{|1^M\rangle_A}|1^L\rangle_B$};
\draw[->] (0,0)--++(0,-\dl);
\node[right] at (0,-\dl/2) {Single qubit rotation $S_q$};
\end{scope}
\begin{scope}[yshift={-5*\step}]
\node [above] at (0,0) {$\alpha|0\rangle_q|0^M\rangle_A|0^L\rangle_B+\beta|1\rangle_q|1^M\rangle_A|1^L\rangle_B$}; 
\end{scope}
\draw[gray!70] (-3.5,1) rectangle (9.1,-7.5);
\end{tikzpicture}
\caption{Sequence of operations in the new GHZ encoding protocol, with the target state (i.e. ignoring errors in the QSP read/write step) for each step shown. The changes of the state in each step are shown in red color.}\label{fig:protocol}
\end{figure}

The following protocol is also described in the chart in Figure~\ref{fig:protocol}.
\begin{itemize}[leftmargin=*,itemsep=5pt]
\item Set-up: We start with a control qubit which is in state $\alpha|0\rangle_q+\beta|1\rangle_q$. We divide the remaining qubits into two registers $A$ and $B$, with $|A|=M$ and $|B|=L$, both $M,L=\Theta(N)$, and $1+M+L=N$. 
The starting states for the qubits in these subsystems are $|0^M\rangle_A$ and $|0^L\rangle_B$ respectively. The total space is $\Lambda:=\{q\}\cup A\cup B$.

\item TAT packet separation on $A$: The set $A$ will undergo the same TAT-rotate-TAT procedure in the first part of the protocol from \cite{yin2025fast}, though squeezed to nearly the extreme limit. The set $B$ will act as a write-space for the QSP method. It will be written with all 1s or all 0s based on the $Y$ sign of qubits in $A$, and then will be used to copy 1s back to $A$ to form the $|1^N\rangle$ term. 
This avoids using a refocusing procedure to transform the state in $A$ into the GHZ state.

Let $X:=\sum_{i\in A}X_i$, $Y:=\sum_{i\in A}Y_i$, and $Z:=\sum_{i\in A}Z_i$.
Recall the TAT packet separation protocol consists of the 3 steps TAT squeezing, controlled $X$ rotation, then reverse TAT squeezing. 
Similar to \eqref{eqn:tatv}, we will take the TAT packet separation operator
\begin{align}\label{eqn:tat}
V_M:=e^{is_2K/(M+1)}e^{i\phi Z_qX/2}e^{-is_1K/(M+1)},
\end{align}
where $K=\frac12(XY+YX)$ as before,
but with different parameters
\begin{align}\label{eqn:svals}
s_1=\frac12\log\frac{M^{3/2}g}{\kappa(M+1)\log M}, \quad \phi=\kappa\frac{\log M}{M}, \quad s_2=\frac12\log\frac1\phi, 
\end{align}
where $g$ will be chosen $\Theta(\sqrt{\log\log M})$, and $\kappa\ge1$ is a choice of fixed numerical constant like $\kappa=1$ or $10$. The choice of $\kappa$ is solely for producing parameters which also work well for small or finite-$M$ numerics; any fixed choice of $\kappa$, such as $\kappa=1$, is sufficient for the asymptotic theorem with $M\to\infty$.
Note that $s_1\approx\frac14\log\frac{Mg}{\kappa\log M}\approx\frac14(\log M)(1+o(1))$, which corresponds to approaching extreme squeezing ($\tau_1\approx s_1/(2\log M)\to \tau_{\mathrm min}\approx1/8$ in the notation of \cite{yin2025fast}).
The extreme squeezing allows for taking the smaller $X$ rotation duration $\phi=\log M/M$, rather than the $\Theta((\log M)^2/M)$ of \cite{yin2025fast}. Since the extremely squeezed packets have width $\Delta y\approx O(1/M)$, shifting them apart by $\phi=\log M/M$ separates them.
While the resulting packets with near extreme squeezing are difficult to fully control or refocus, it will not matter for this protocol. In particular, we will only need to know where the main mass, or ``core'', of the packet goes, and will not need to control the tail behavior.

Applying $V_M$ to the qubits in $A$ with $q$ as the control qubit, the resulting state in $A$ consists of two packets $|\chi_\pm\rangle_A$, one close to each of the $Y$ poles in the Dicke sphere.
These states $|\chi_\pm\rangle_A$ may be very non-Gaussian due to the extreme squeezing, but their structure will not matter. It only matters that the cores of $|\chi_\pm\rangle_A$ are near enough to opposite $Y$ poles, in particular separated by a gap from each other. The total state is
\begin{align}
\alpha|0\rangle_q|\chi_+\rangle_A|0^L\rangle_B+\beta |1\rangle_q|\chi_-\rangle_A|0^L\rangle_B.
\end{align}

\item QSP read/write: Since the packets $|\chi_\pm\rangle_A$ are mostly located near the $Y$ poles, $Y$ will be of order $\Theta(M)$ on each. By using $Y$ in the generator for a unitary, this can then be used to flip all the bits in $|0^L\rangle_B$ in time close to order $\approx 1/M$, or more precisely $O(\log M/M)$. This is in contrast to the time $\Theta(1)$ required for a standard controlled rotation, which uses $\pm1$ control values.
The bit flips in time $O(\log M/M)$ are implemented using a quantum signal processing (QSP) read/write method to handle the variation in $Y$ values; in particular, QSP can be used to approximately implement the $\sgn$ function and coherently write properties of the quantum state into the qubits in $B$ \cite{low2017hamiltonian,gilyen2019quantum,haah2019product,rall2021faster,liu2025toward}.
This read/write process is described by the unitary $W_{A:B}$ defined following Theorem~\ref{thm:qsp}.
The result is a state which is approximately
\begin{align}
\alpha|0\rangle_q|\chi_+\rangle_A|0^L\rangle_B+\beta \reds{i^L}|1\rangle_q|\chi_-\rangle_A\reds{|1^L\rangle_B},
\end{align}
with the red text indicating the changes due to the QSP read/write method.

\item The inverse TAT packet separation $V_M^\dagger$ can then be applied to $\{q\}\cup A$ to send everything in $A$ (i.e. $|\chi_\pm\rangle_A$) back to $|0^M\rangle_A$, up to  the error from applying $W_{A:B}$. 

\item Controlled rotation, phase rotation: using the block $|1^L\rangle_B$, one can perform the controlled rotation $R_A=\exp(i\pi X/4)\exp(-i\pi Z_BX/(4L))$, where $Z_B:=\sum_{b\in B}Z_b$, to flip all the bits $|0^M\rangle_A\to|1^M\rangle_A$ in time $O(1/L)$, and then use a single qubit rotation $S_q$ to remove the remaining $i^{L+M}$ phase.
\end{itemize}

For convenience, we collect here the definitions of unitary operations in the protocol, except for the QSP read/write unitary $W_{A:B}$ which will be defined in \eqref{eqn:W-write} following Theorem~\ref{thm:qsp}:
\begin{align}\label{eqn:terms}
\begin{aligned}
V_M&= e^{is_2K/(M+1)}e^{i\phi Z_qX/2}e^{-is_1K/(M+1)}, \quad s_1,\phi,s_2\text{ as in \eqref{eqn:svals}}\\
R_A&= e^{i\pi X/4}e^{-i\pi Z_BX/(4L)},\\
S_q&= e^{-i\pi\xi(1-Z_q)/4},\quad \xi=(L+M)\mod 4.
\end{aligned}
\end{align}

Theorem~\ref{thm:encoding} is implied by
\begin{thm}[main result]\label{thm:protocol}
The protocol
\begin{align}\label{eqn:U-protocol}
U_N&= S_q R_A V_M^\dagger W_{A:B}V_M
\end{align}
described above and in Figure~\ref{fig:protocol} takes time $O(\log N/N)$ and produces a $(1-o(1))$-fidelity GHZ encoding. The Hamiltonian generator for each step is 2-local with bounded 2-qubit interaction terms. 
The $o(1)$ error term in the fidelity can be taken as $O\left(\frac{(\log\log M)^{3}}{(\log M)^4}\right)$, uniformly over $|\alpha|^2+|\beta|^2=1$.
\end{thm}

Since $L,M=\Theta(N)$, it will be clear from \eqref{eqn:terms} and the later definition of $W_{A:B}$ that the total evolution time is $O(\log N/N)$ and that all Hamiltonian generators are 2-local with operator-norm-bounded 2-qubit interaction terms.
The main work will be to show that the protocol produces a $(1-o(1))$-fidelity GHZ state.

\subsection{Outline and notation}\label{subsec:outline}
The rest of the paper is organized as follows.
\begin{itemize}
\item Section~\ref{sec:proof}: Proof overview, including additional background and references.
\item Sections~\ref{sec:tat-proof} and \ref{sec:dc-error}: Proof of TAT packet separation Theorem~\ref{thm:tat}.
\item Section~\ref{sec:qsp-proof}: Proof of QSP method Theorem~\ref{thm:qsp}.
\end{itemize}

For convenience, we collect here some parameters defined in the text. 
\begin{align*}
\def\arraystretch{1.5}
\begin{array}{ll}
\text{Context} &\text{Parameters}\\\hline
\text{TAT packet separation} & s_1=\frac12\log\frac{M^{3/2}g}{\kappa(M+1)\log M}, \quad \phi=\kappa\frac{\log M}{M}, \quad s_2=\frac12\log\frac1\phi\\
\text{Continuum scaling}&h=\frac{2}{M+1},\quad y_m=hm\\
\text{Gaussian packet core} & \sigma_0=\frac{\sqrt{M}}{M+1},\quad a\in(0,1/2),\quad b=1/2,\quad g=\sqrt{8\log\log M}/a\\
\text{Packet squeezing} & \sigma_*:=e^{-2s_1}\sigma_0=\kappa\frac{\log M}{Mg},\quad e^{2s_2}\sigma_*=1/g
\end{array}
\end{align*}
The parameter $\kappa\ge1$ is a fixed constant, such as $\kappa=1$ or $10$, and thus may sometimes be silently absorbed into asymptotic notation. As previously mentioned, its choice is only for producing parameters which also work for small or finite-$M$ numerics; any fixed choice, such as $\kappa=1$ is sufficient for the asymptotic theorem with $M\to\infty$.

Notation: We use standard asymptotic notation $o,O,\Theta,\Omega,\omega$. Generic constants $C,c$ may change from line to line.
Throughout, we will use the coordinate $y\in[-1,1]$ to denote a normalized $Y$-coordinate, which we can think of as living on the Dicke sphere. For the $Y$-Dicke state $|m\rangle_Y$, $m=-M/2,-M/2+1,\ldots,M/2$, the associated $y$-coordinate is $y=2m/M\in[-1,1]$. For a continuum packet approximation, we will use the centered coordinate $y=y_m:=2m/(M+1)$, which is asymptotically equivalent as it differs only by $O(1/M)$.

\section{Proof overview}\label{sec:proof}

In this section, we give an overview of the proof of Theorem~\ref{thm:protocol}, including the main intermediate results, which we prove in Sections~\ref{sec:tat-proof}--\ref{sec:qsp-proof}.

\begin{proof}[Proof outline of Theorem~\ref{thm:protocol}]
The main technical result we need to prove Theorem~\ref{thm:protocol} is the following quantitative estimate for the behavior under TAT packet separation $V_M$.
\begin{thm}[TAT packet separation to $Y$ poles]\label{thm:tat}
Let
\begin{align}\label{eqn:VM}
V_M^\pm:=e^{is_2K/(M+1)}e^{\pm i\phi X/2}e^{-is_1K/(M+1)},
\end{align}
for $K=\frac12(XY+YX)$ and $s_1, \phi, s_2$ as in \eqref{eqn:svals}.
Then defining projections $\Pi_{+,\gamma}:=\oneb_{Y\ge\gamma M}$ and $\Pi_{-,\gamma}:=\oneb_{Y\le-\gamma M}$, we have
\begin{align}\label{eqn:ypoles}
\left\|(1-\Pi_{\pm,1/3})V_M^\pm|0^M\rangle_A\right\|_{\ell^2}&=O\left(\frac{(\log\log M)^{3/2}}{(\log M)^2}\right).
\end{align}
\end{thm}

Theorem~\ref{thm:tat} ensures that the TAT packet separation $V_M^\pm$ results in states $V_M^\pm|0\rangle^{\otimes M}$ which are sufficiently concentrated in a region near opposite $Y$ poles.
The resulting packets $V_M^\pm|0\rangle^{\otimes M}$ may be highly non-Gaussian and have complicated or stringy tails, but all that matters is that the vast majority of the mass is close enough to the $Y$ poles; in particular, there is a macroscopic $y$-coordinate gap between the cores of the packets $V_M^\pm|0\rangle^{\otimes M}$.
Note it is not enough to just control the center location of the packets, or to control the packets in a small region near the initial pole. We will need to control the evolution of the entire Gaussian packet ``core'' to macroscopic distances in $y$.

The full proof of Theorem~\ref{thm:tat} will be given in Sections~\ref{sec:tat-proof} and \ref{sec:dc-error}. The main steps for the proof of Theorem~\ref{thm:tat} are as follows. 
\begin{enumerate}[leftmargin=*,itemsep=5pt]
\item Gaussian (continuum) approximation to $|0\rangle^{\otimes M}$: In the $Y$-Dicke basis $|m\rangle_Y$, for $m=-M/2,-M/2+1,\ldots,M/2$, with $Y|m\rangle=2m|m\rangle$,
\begin{align}\label{eqn:0gaussian}
|0\rangle^{\otimes M}&=\sum_{m=-M/2}^{M/2}2^{-M/2}\binom{M}{M/2+m}^{1/2}|m\rangle_Y.
\end{align}
From standard central limit theorem/large deviation estimates, this will be well-approximated by a Gaussian core $f_{M,g}$, where by core we mean we start with a Gaussian function and use a cutoff function to exclude the infinite tails. The core $f_{M,g}$ will include around $g=\Theta(\sqrt{\log\log M})$ standard deviations. Importantly, $g\to\infty$, but only very slowly. While we already know \eqref{eqn:0gaussian} looks like a centered Gaussian with width $1/\sqrt{M}$ from the semiclassical approximation, we will need more specific norm, tail, and derivative estimates for the rest of the proof.

\item Continuum evolution to $Y$ poles under TAT packet separation $V_M$: We can identify the continuum limit generators $L_{v_0},L_{v_1}$ for each of the 3 steps in \eqref{eqn:VM}. Using the finite-difference representation of the dynamics to relate to a classical flow, such as the TAT flow studied in \cite{munoz2023phase}, we will prove that the Gaussian core $f_{M,g}$ is mapped by the (continuum version of) TAT packet separation to $y$-coordinates within $1/3$ of a $Y$ pole.
More precisely, letting $T_M^\pm:=e^{s_2L_{v_1}}e^{\pm\phi L_{v_0}}e^{-s_1L_{v_1}}$ be the continuum analogues of $V_M^\pm$, we will show
\begin{align}\label{eqn:supp}
\supp(T_M^+ f_{M,g})\subset\{y>1/3\},\quad\text{and}\quad \supp(T_M^- f_{M,g})\subset\{y<-1/3\}.
\end{align}
So the continuum evolution of the initial Gaussian state satisfies exactly the desired property in Theorem~\ref{thm:tat}.
Notably, this step uses a continuum limit, but does not use the small $y$ Holstein--Primakoff approximation \cite{holstein1940field}  at the starting $+Z$ pole, which is not applicable for macroscopic $y$ locations. In the next step, we will have to estimate the error from the continuum limit approximation.

\item Discrete-continuum error bounds: This is the most technical step, and is responsible for the main part of the error in \eqref{eqn:ypoles}. 
We need to estimate the difference between the discrete evolution and the continuum evolution of the Gaussian packet. This will be done using Duhamel's principle to bound the evolution difference in terms of the difference between the discrete and continuum generators. The latter will then be bounded using Taylor expansion and a sequence of ordinary differential equation (ODE) transport estimates. 

Heuristically, the approximation errors depend on both how well the discrete approximation at scale $h=O(1/M)$ compares to the packet width, and how large the derivative of the evolution operator is at that time. More sample points $y_m=hm$ per packet width imply a better approximation. Control on the derivative of the evolution operator is also necessary to bound the error.

\end{enumerate}

After the TAT packet separation, the state is
\begin{align}
\alpha|0\rangle_q{|\chi_+\rangle_A}|0^L\rangle_B+\beta|1\rangle_q{|\chi_-\rangle_A}|0^L\rangle_B,
\end{align}
for states $|\chi_\pm\rangle_A=V_M^\pm|0^M\rangle$.
By Theorem~\ref{thm:tat}, $|\chi_+\rangle_A$ and $|\chi_-\rangle_A$ lie nearly in the subspaces $\{Y\ge M/3\}$ and $\{Y\le -M/3\}$, respectively.
The large $Y=\Theta(M)$ from the $A$-qubits will be used to perform qubit rotations in $B$ in time $O(\log M/M)$ using a 2-local Hamiltonian with bounded 2-qubit interaction terms. The $Y$-separation between the two packets $|\chi_\pm\rangle_A$ allows to perform a controlled-like rotation on $B$, by approximating the $\sgn$ function using a construction from quantum signal processing (QSP) \cite{low2017hamiltonian,gilyen2019quantum,haah2019product}.
QSP allows for approximating general functions in terms of products of unitary matrices, and is thus particularly applicable to coherently reading out properties of quantum states \cite{liu2025toward,rall2021faster}; in this case, the $Y$-sign.

\begin{thm}[QSP $\sgn$ approximation]\label{thm:qsp}
Consider the single-qubit Hilbert space $\C^2$, with identity $I_2$ and single-qubit-Pauli $X$.
Let $0<\gamma<1$ and $0<\delta<1/2$. There is a constant $C_\gamma$, a single qubit unitary $U_0$, and rank-1 projectors $P_1,\ldots,P_d$, for $d\le C_\gamma\log(1/\delta)$, such that
\begin{align}\label{eqn:Qdelta}
Q_\delta(x)=U_0\prod_{j=1}^{d}e^{i\pi x(2P_j-1)/4}
\end{align}
satisfies
\begin{align}
\sup_{x\in[\gamma,1]}\|Q_\delta(x)-I_2\|\le\delta,\quad \sup_{x\in[-1,-\gamma]}\|Q_\delta(x)-iX\|\le\delta.
\end{align}
\end{thm}

The proof uses construction of a polynomial approximation $p$ to the $\sgn$ function of degree $O(\eta^{-1}\log(1/\delta))$ such that \cite{eremenko2007uniform,low2017hamiltonian,gilyen2019quantum}
\begin{align}
|p(x)-\sgn(x)|\le \delta\text{ for }|x|\ge \eta.
\end{align}
Constructions and properties from \cite{haah2019product} are then used to turn this into the desired unitary operator $Q_\delta$. The proof of Theorem~\ref{thm:qsp} is given in Section~\ref{sec:qsp-proof}.

We can apply Theorem~\ref{thm:qsp} to construct the unitary $W_{A:B}$ which writes the sign of $Y=\sum_{i\in A}Y_i$ to the workspace $B$ qubits in time $O(\log M/M)$ with small error.
For $j=1,\ldots,d$, write $2P_j-1=\mathbf n_j\cdot\boldsymbol\sigma$ for some $\|\mathbf n_j\|_2\le1$. Let $\boldsymbol\sigma_b=(X_b,Y_b,Z_b)$ act on qubit $b\in B$, and define
\begin{align}
H_j&=Y\sum_{b\in B}\mathbf n_j\cdot\boldsymbol\sigma_b,
\end{align}
which has bounded 2-qubit interaction terms. 
Define
\begin{align}\label{eqn:W-write}
W_{A:B}:=U_0^{\otimes L}e^{i\pi H_1/(4M)}e^{i \pi H_2/(4M)}\cdots e^{i\pi H_d/(4M)}=Q_\delta^{\otimes L}\left(\frac{Y}{M}\right),
\end{align}
where by the notation $Q_\delta^{\otimes L}(Y/M)$, we mean we apply the same $Q_\delta(Y/M)$ to each of the $L$ qubits in $B$ (in parallel). This takes total time $\frac{d\pi}{4M}=O(\frac{\log(1/\delta)}{M})$. We will take $\delta=1/M^2$, so the time for $W_{A:B}$ will be $O(\log M/M)$.

To complete the proof of Theorem~\ref{thm:protocol}, we combine Theorems~\ref{thm:tat} and \ref{thm:qsp}, the latter with $\gamma=1/3$. 
On the range of the projection $\Pi_{+,1/3}\otimes I_B$, we have $Q_\delta(Y/M)=I_2+O(\delta)$. Since $\|Q_\delta(x)^{\otimes L}-I^{\otimes L}\|\le L\|Q_\delta(x)-I\|\le L\delta$, we get
\begin{align*}
W_{A:B}V_M|0\rangle_q|0^M\rangle_A|0^L\rangle_B&=W_{A:B}\Pi_{+,1/3}V_M^+|0\rangle_q |0^M\rangle_A|0^L\rangle_B+O\left(\frac{(\log\log M)^{3/2}}{(\log M)^2}\right)\\
&=|0\rangle_q|\chi_+\rangle_A|0^L\rangle_B+O(L\delta)+O\left(\frac{(\log\log M)^{3/2}}{(\log M)^2}\right),\numberthis
\end{align*}
where the errors are in norm bound.
Similarly, on the range of $\Pi_{-,1/3}$, we have $Q_\delta(Y/M)=iX+O(\delta)$, so
\begin{align*}
W_{A:B}V_M|1\rangle_q|0^M\rangle_A|0^L\rangle_B&=i^L|1\rangle_q|\chi_-\rangle_A|1^L\rangle_B+O(L\delta)+O\left(\frac{(\log\log M)^{3/2}}{(\log M)^2}\right).\numberthis
\end{align*}
The remaining operations $V_M^\dagger$, $R_A$, and $S_q$ in the protocol in Theorem~\ref{thm:protocol} are unitary and exact, so the error bound carries forward, and we obtain
\begin{align}\label{eqn:psi-out}
U_N(\alpha|0\rangle_q+\beta|1\rangle_q)|\mathbf 0\rangle_{A\cup B} &= (\alpha|\mathbf0\rangle_\Lambda+\beta|\mathbf1\rangle_\Lambda)+O(L\delta)+O\left(\frac{(\log\log M)^{3/2}}{(\log M)^2}\right).
\end{align}
Since we take $\delta=1/M^2$, this gives GHZ encoding fidelity $1-O\left(\frac{(\log\log M)^{3}}{(\log M)^4}\right)$ as follows. Letting $|\psiout\rangle:=U_N(\alpha|0\rangle_q+\beta|1\rangle_q)|\mathbf 0\rangle_{A\cup B}$, and $|G_{\alpha\beta}\rangle:=\alpha|\mathbf 0\rangle_\Lambda+\beta|\mathbf 1\rangle_\Lambda$ be the GHZ-encoding target, then \eqref{eqn:psi-out} gives $|\psiout\rangle=|G_{\alpha\beta}\rangle+|R\rangle$ with $\|R\|_2\le O\left(\frac{(\log\log M)^{3/2}}{(\log M)^2}\right)$. Then the fidelity is
\begin{align*}
|\langle\psiout|G_{\alpha\beta}\rangle|^2=|1+\langle R|G_{\alpha\beta}\rangle|^2&\ge 1+2\Re\langle R|G_{\alpha\beta}\rangle+(\|R\|_2^2-\|R\|_2^2)\\
&=\||G_{\alpha\beta}\rangle+|R\rangle\|_2^2-\|R\|_2^2\\
&=\|\psiout\|_2^2-\|R\|_2^2\ge1-O\left(\frac{(\log\log M)^{3}}{(\log M)^4}\right),\numberthis
\end{align*}
uniformly over $|\alpha|^2+|\beta|^2=1$.
The physical implementation time is $O(\log M/M)$ for $V_M$ and $V_M^\dagger$, $O(\frac{\log(1/\delta)}{M})=O(\frac{\log M}{M})$ for $W_{A:B}$, and $\pi/(4L)=O(1/M)$ for $R_A$, so we obtain total encoding time $O(\log M/M)$. 
\end{proof}

\section{First parts of the proof of Theorem~\ref{thm:tat}: Continuum limit and evolution}\label{sec:tat-proof}

In this section, we prove the first two steps of the proof of Theorem~\ref{thm:tat} listed in Section~\ref{sec:proof}, namely
\begin{enumerate}
\item Gaussian (continuum) approximation to $|0\rangle^{\otimes M}$.
\item Continuum evolution to $Y$ poles under TAT packet separation $V_M$.
\end{enumerate}
The main results for these are stated as Propositions~\ref{prop:gstate} and \ref{prop:supp}, respectively.

We first briefly discuss some set-up, continuum limit scaling, and notation conventions. 
On $A$, the initial state $|0^M\rangle_A$ is symmetric in the qubits in $A$, and all operators in \eqref{eqn:svals} are also symmetric in $A$. Thus we work in the symmetric subspace $\mathrm{Sym}^M(\C^2)$, which has dimension $M+1$.
It will be convenient to work in the $Y$-Dicke basis $|m\rangle_Y$, $m=-M/2,-M/2+1,\ldots,M/2$, which for $J_y:=Y/2$ satisfy
\begin{align}
J_y|m\rangle_Y=m|m\rangle_Y.
\end{align}
Explicitly, $|m\rangle_Y$ is the equal superposition of the product states formed from the $\sigma_y$-basis states $|\pm\rangle=\frac1{\sqrt{2}}(|0\rangle\pm i|1\rangle)$, with $(M/2+m)$ $|+\rangle$ factors and $(M/2-m)$ $|-\rangle$ factors, i.e.~letting $|\pm\rangle_S:=\otimes_{i\in S}|\pm\rangle_i$,
\begin{align}
|m\rangle_Y&=\binom{M}{p}^{-1/2}\sum_{\substack{S\subseteq A\\|S|=p}}|+\rangle_S|-\rangle_{A\setminus S},\quad p=m+\frac{M}{2}\in\Z.
\end{align}

As in Figures~\ref{fig:bloch3} and \ref{fig:bloch}, states in the symmetric subspace (Dicke manifold) can be visualized on the Dicke sphere using spin-coherent states.
Recall a spin-coherent state \cite{arecchi1972atomic} is a state of the form $|\mathbf n\rangle^{\otimes M}$ for a single-qubit state $|\mathbf n\rangle$ on the Bloch sphere, i.e.~all the spins point in the same direction $\mathbf n\in \Sb^2$. 
The distribution of a state $|\psi\rangle$ can be visualized on the unit sphere via its Husimi-$Q$ function on the sphere, defined as $Q_\psi(\mathbf n):=|\langle\mathbf n^{\otimes M}|\psi\rangle|^2$. A spin-coherent state appears as a small concentrated region centered at $\mathbf n$, while a $Y$-Dicke basis state $|m\rangle_Y$ appears as a band around a constant $y$-value, $y\approx 2m/M$.

If we consider the Dicke sphere $y$-coordinate $y\in[-1,1]$, we can partition $[-1,1]$ into $M+1$ equal-sized intervals, and associate each with one of the $M+1$ basis states $|m\rangle_Y$. Since $Y=2J_y$, we can think of the relation between $y$ and $m$ as $y=\frac{2m}{M+1}$, which relates each Dicke state $|m\rangle_Y$ with the center of a size $1/(M+1)$ cell in $[-1,1]$. The lattice spacing between adjacent states is then $h=\frac{2}{M+1}$, and we have $y$-coordinates $y_m=hm$.

\subsection{Gaussian (continuum) approximation to \texorpdfstring{$|0\rangle^{\otimes M}$}{|0>^M}}\label{subsec:g0}

Recall from \eqref{eqn:0gaussian}, or by computing ${}_Y\!\langle m|0^M\rangle$, the state $|0\rangle^{\otimes M}$ is expressed in the $Y$-Dicke basis $|m\rangle_Y$ as
\begin{align}\label{eqn:psi-M}
|\psi_M\rangle=\sum_{m=-M/2}^{M/2}2^{-M/2}\binom{M}{M/2+m}^{1/2}|m\rangle_Y,
\end{align}
where we use $|\psi_M\rangle$ to mean the vector is always written in the $|m\rangle_Y$ basis, and we view it only in terms of its coordinates in this basis.
We will show this is well-approximated by a Gaussian core $f_{M,g}$, which is centered at $y=0$ and has width $\Theta(1/\sqrt{M})$. 
This is in agreement with the large-$M$ semiclassical approximation, but we will need norm, tail, and derivative estimates which are not covered by the usual semiclassical approximation.

Let $F_g(u):=\chi(u/g)\frac{e^{-u^2/4}}{(2\pi)^{1/4}}$, with $\chi$ a smooth $C^\infty$ cut-off function which is 1 on $|x|\le a$ and $0$ for $|x|\ge b$, for 
\begin{align}\label{eqn:gab}
g=\sqrt{8\log\log M}/a, \quad \text{any fixed }0<a<1/2, \quad b=1/2.
\end{align}
The function $F_g$ is the ``core'' of a Gaussian which is normalized as a half-density, and the scale $g$ ensures that the core slowly covers more and more standard deviations as $M\to\infty$. We use a Gaussian core instead of the full Gaussian because we want to ignore tail behavior, which can be poorly behaved after extreme squeezing in $V_M$.
We then define
\begin{align}\label{eqn:fm}
f_{M,g}(y):=\sigma_0^{-1/2}F_g(y/\sigma_0),
\end{align}
for $\sigma_0:=\sqrt{M}/(M+1)$.\footnote{This choice of $\sigma_0$ is asymptotically the same as $1/\sqrt{M}$, which is the initial semiclassical packet width of $|0\rangle^{\otimes M}$. We take this particular form $\sigma_0=\sqrt{M}/(M+1)$ because it relates more conveniently to the precise discrete distribution of $|0\rangle^{\otimes M}$. In particular, the continuum limit terms $q_M(m)$ in the proof take particularly nice forms with this choice of $\sigma_0$, even though initially is looks more complicated than using $1/\sqrt{M}$.} 
Then $f_{M,g}$ is supported on $|y|/\sigma_0\le bg$.
Note the square root normalization in \eqref{eqn:fm} is because we work with half-densities and $L^2$ normalization.
We will show
\begin{prop}[Gaussian state approximation]\label{prop:gstate}
Let $f_{M,g}$ be as in \eqref{eqn:fm}, with $g=\Theta(\sqrt{\log\log M})$.
Letting $(S_hf)_m:=\sqrt{h}f(hm)$ be the discretization of a function $f:\R\to\C$ at scale $h=2/(M+1)$, then
\begin{align}\label{eqn:gaussian-state}
\|\psi_M-S_hf_{M,g}\|_{\ell^2}&\le O\left(\frac{g^4}{M}\right)+O(e^{-a^2g^2/4}).
\end{align}
\end{prop}
\begin{proof}
Let $p_M(m):=2^{-M}\binom{M}{m+M/2}$ be the square of the entries of $|\psi_M\rangle$. Let $\tilde f_{M,g}:=\sigma_0^{-1/2}G(y/\sigma_0)$ for $G(u):=\frac{e^{-u^2/4}}{(2\pi)^{1/4}}$, so this is essentially $f_{M,g}$ but without the cutoff $\chi$. Letting $q_M(m):=|(S_h\tilde f_{M,g})_m|^2$, we can evaluate
$q_M(m)=\sqrt{\frac{2}{\pi M}}e^{-2m^2/M}$.
To estimate
\begin{align}
\|\psi_M-S_hf_{M,g}\|_{\ell^2}^2&=\sum_{m=-M/2}^{M/2}\left[\sqrt{p_M(m)}-\chi\left(hm/(g\sigma_0)\right)\sqrt{q_M(m)}\right]^2,
\end{align}
we first use Stirling's formula $\log n!=n\log n-n+\frac12\log(2\pi n)+O(1/n)$ to compare $p_M$ to $q_M$. For $m=o(M)$, we have
\begin{align*}
\log p_M(m)&=-M\log2+\log\frac{M!}{(\frac{M}{2}+m)!(\frac{M}{2}-m)!}\\
&=\frac12\log\frac{2}{\pi M}-\frac{2m^2}{M}+O\left(\frac{1}{M}+\frac{m^2}{M^2}+\frac{m^4}{M^3}\right),\numberthis
\end{align*}
where we expanded $\log(\frac{M}{2}\pm m)=\log\left(\frac{M}{2}(1\pm\frac{2m}{M})\right)$ to third order.
If we restrict to $|m|\le bg\sqrt{M}$ and recall $h=2/(M+1)$, then we get uniformly over such $m$,
\begin{align}\label{eqn:logpq}
\log\frac{p_M(m)}{q_M(m)}&= O\left(\frac{1}{M}+\frac{g^2}{M}+\frac{g^4}{M}\right)=O\left(\frac{g^4}{M}\right).
\end{align}
Letting $q_M=p_Me^{-\delta}$ with $\delta(m)=O(g^4/M)=o(1)$ by \eqref{eqn:logpq}, and noting $\|\sqrt{p_M}\|_{\ell^2}=1$, we see
\begin{align}
\|\oneb_{|m|\le bg\sqrt{M}}(\sqrt{p_M}-\sqrt{q_M})\|_{\ell^2}&=O\left(\frac{g^4}{M}\right).
\end{align}

It remains to handle the tails and cutoff $\chi$. 
Recall $\chi$ is $1$ on $|x|\le a$ and $0$ for $|x|\ge b$; thus $\chi(hm/(g\sigma_0))=1$ for $|m|\le ag\sigma_0/h=ag\sqrt{M}/2$, and so it will suffice to consider $|m|\ge ag\sqrt{M}/2$.
By \eqref{eqn:logpq}, we have $q_M(m)=O(p_M(m))$ for $ag\sqrt{M}/2\le |m|\le bg\sqrt{M}$. Since also $\chi(hm/(g\sigma_0)=0$ for $|m|>bg\sigma_0/h= bg\sqrt{M}/2$, it suffices to bound $p_M$. 
One can quickly check, for either parity of $M$, that the probability distribution with discrete density $p_M$ is that of $S_M:=\xi_1+\cdots+\xi_M$, for i.i.d. symmetric $\xi_j\in\{\pm1/2\}$.
Hoeffding's inequality then gives
\begin{align}
\sum_{|m|\ge r}p_M(m)=\P[|S_M|\ge r]\le 2e^{-2r^2/M},
\end{align}
and taking $r=ag\sqrt{M}/2$ gives $\|\oneb_{|m|\ge ag\sqrt{M}/2}(\sqrt{p_M}-\chi\sqrt{q_M})\|_{\ell^2}=O(e^{-a^2g^2/4})$. 
\end{proof}

\subsection{Finite-difference relation for the dynamics} 

In order to consider the continuum time evolution of the packets to the $Y$ poles under TAT packet separation $V_M$, we first identify the continuum limits for the generators $iX/2$ and $iK/(M+1)$ for $K=\frac12(XY+YX)$. 
This has been well studied; for example, for $K$ this corresponds to the classical TAT flow studied in \cite{munoz2023phase}.
For completeness, we give a derivation for the normalization and form needed here.
Note that both $X$ and $K$ are tri-diagonal in the $Y$-Dicke basis, since they only act between neighboring states $|m\rangle_Y$ and $|m\pm1\rangle_Y$. We can then write them as finite-difference operators, and identify the continuum limit.
\begin{lem}[finite-difference relation]\label{lem:Dhv}
For the scale $h=2/(M+1)$ and function $v:[-1,1]\to\R$, define the finite-difference operator 
\begin{align}\label{eqn:Dhv}
(D_{h,v}\psi)_m:=\frac{v(hm-h/2)}{2h}\psi_{m-1}-\frac{v(hm+h/2)}{2h}\psi_{m+1},\quad m=-M/2,-M/2+1,\ldots,M/2,
\end{align}
which is skew-adjoint. We set $\psi_m=0$ for $m$ outside the above range.
Let $X=\sum_{j=1}^M X_j$, $Y=\sum_{j=1}^M Y_j$, and $K=\frac12(XY+YX)$.
Then in the $Y$-Dicke basis $|m\rangle_Y$ on the symmetric subspace,
\begin{align}\label{eqn:Dv}
\begin{aligned}
\frac{iX}{2}&=D_{h,v_0},\quad\text{ for } v_0(y)=\sqrt{1-y^2},\\
\frac{iK}{M+1}&=D_{h,v_1},\quad\text{ for } v_1(y)=2y\sqrt{1-y^2}.
\end{aligned}
\end{align}
\end{lem}
\begin{rmk}\label{rmk:Lv}
For differentiable $v$, as $h\to0$, we can informally check the finite difference operator $D_{h,v}$ should approach the continuum operator
\begin{align}\label{eqn:Lv}
L_v&:=-v(y)\partial_y-\frac12v'(y),\quad y\in(-1,1).
\end{align}
We derive this continuum target in \eqref{eqn:DL0} at the end of the proof of the lemma below.
In Section~\ref{subsec:dc-error-proof} (Lemma~\ref{lem:DL-approx}), we will prove a quantitative bound on closeness of $D_{h,v}$ to $L_v$; equation~\eqref{eqn:Lv} is not meant at the moment to be a formal limiting statement about $D_{h,v}$.
Note also that $L_v$ is skew-hermitian, just like $D_{h,v}$, on $C_c^\infty(-1,1)$, though it may not be skew-adjoint.
\end{rmk}

\begin{proof}
Let $|\pm\rangle=\frac{1}{\sqrt{2}}(|0\rangle\pm i|1\rangle)$ be the eigenstates of single-qubit Pauli $Y_j=\sigma_y$, so that on a single qubit $j$, $X_j|\pm\rangle_j=\pm i|\mp\rangle$. 
In the $Y$-Dicke basis $|m\rangle_Y=\binom{M}{p}^{-1/2}\sum_{\substack{S\subseteq A\\|S|=p}}|+\rangle_S|-\rangle_{A\setminus S}$, where $p=m+M/2$, we see $X=\sum_{j=1}^MX_j$ is tri-diagonal in this basis since it only flips one $\pm$ at a time. Similarly, since $Y|m\rangle_Y=2m|m\rangle_Y$, $K$ is also tri-diagonal. Also, since $X$ always changes the number of $+$'s, we see $\langle m|X|m\rangle=0$.
Since $X=2J_x=\frac1i(J_+^Y-J_-^Y)$ for raising and lowering operators $J_\pm^Y:=J_z\pm iJ_x$, standard angular momentum formulas \cite[Sec. 3.5.3]{sakurai2020modern} (alternatively direct counting using $X|-\rangle=-i|+\rangle$) give
\begin{align*}
\langle m+1|X|m\rangle_Y &=-ic_m,\quad\langle m+1|K|m\rangle_Y=-i(2m+1)c_m,\numberthis
\end{align*}
for $c_m:=\sqrt{(M/2-m)(M/2+m+1)}$.
Letting $y_m=mh$ and recalling $h=2/(M+1)$, we can rewrite $c_m$ in terms of $y_m$ via
\begin{align*}
h^2c_m^2&=h^2\left[\left(\frac{M}{2}\right)^2+\frac{M}{2}-m^2-m\right] =h^2\left[\left(\frac{M+1}{2}\right)^2-(m+1/2)^2\right]=1-(y_m+h/2)^2.
\end{align*}
As $(X\psi)_m=\langle m|X|m-1\rangle\psi_{m-1}+\langle m|X|m+1\rangle\psi_{m+1}$, we then get
\begin{align*}
\left(\frac{iX}{2}\psi\right)_m&=\frac12\left[c_{m-1}\psi_{m-1}-c_m\psi_{m+1}\right]=(D_{h,v_0}\psi)_m,\\
\left(\frac{iK}{M+1}\psi\right)_m&=\frac{2}{M+1}\left[(m-1/2)c_{m-1}\psi_{m-1}-(m+1/2)c_m\psi_{m+1}\right]=(D_{h,v_1}\psi)_m,
\end{align*}
which proves the lemma. 

To see the continuum target is \eqref{eqn:Lv} in Remark~\ref{rmk:Lv}, let $y_m:=hm$ and $\psi(y_m):=\psi_m$. Consider $m=m(h)$ so that $y_m=hm\to y\in[-1,1]$.
Then write
\begin{align*}
(D_{h,v}\psi)_m&=\frac{v(y_m-h/2)}{2h}\left[\psi(h(m-1))-\psi(h(m+1))\right]+\left[\frac{v(y_m-h/2)}{2h}-\frac{v(y_m+h/2)}{2h}\right]\psi(h(m+1))\\
&\xrightarrow{h\to0}-v(y)\psi'(y)-\frac12v'(y)\psi(y),\numberthis\label{eqn:DL0}
\end{align*}
which identifies \eqref{eqn:Lv} as the continuum generator. 
\end{proof}

\subsection{Continuum time evolution to \texorpdfstring{$Y$}{Y} poles under TAT packet separation \texorpdfstring{$V_M$}{V_M}}

In this section, we prove a slightly stronger version of \eqref{eqn:ypoles}:

\begin{prop}[Continuum packet support]\label{prop:supp}
Fix $a\in(0,1/2)$, $\kappa\ge1$, and $b=1/2$.
Let $T_M^\pm:=e^{s_2L_{v_1}}e^{\pm\phi L_{v_0}}e^{-s_1L_{v_1}}$, with $L_v$ and $v_0,v_1$ as in Lemma~\ref{lem:Dhv}, and $e^{sL_v}$ defined as the solution to $\partial_t\psi_t=L_v\psi_t$ in \eqref{eqn:ctransport}. Let $f_{M,g}$ be the Gaussian core defined in \eqref{eqn:fm}. 
Then there are constants $M_0$ and $\epsilon>0$ such that for $M\ge M_0$,
\begin{align}\label{eqn:supp1}
\supp(T_M^+ f_{M,g})\subset\{1/3<y<1-\epsilon\},\quad\text{and}\quad\supp(T_M^- f_{M,g})\subset\{-1+\epsilon<y<-1/3\}.
\end{align}
\end{prop}

The main property used to prove the lemma is that we can explicitly evaluate the solution $(e^{tL_v}f)(y)$ for long times, up until a branch point at $y=\pm1$ in the solution. In particular, we can evaluate the time evolved state exactly well beyond the small $y$ approximation regime.
Avoiding the poles $y=\pm1$ is the reason for choosing $s_2$ in the protocol so that the support in \eqref{eqn:supp1} is bounded away from the actual $Y$ pole (cf.~Figure~\ref{fig:bloch}).
The calculations of the flows below are standard; for example, compare with the classical TAT dynamics along the separatrix in \cite{munoz2023phase}.

Let $\Phi_v^t$ be the classical flow satisfying the ordinary differential equation
\begin{align}\label{eqn:char}
\partial_t\Phi_v^t(y)=v(\Phi_v^t(y)),\quad \Phi_v^0(y)=y.
\end{align}
Recalling $L_v=-v\partial_y-\frac12v'$ and using the method of characteristics, we find the solution
\begin{align}
(e^{tL_v}f)(y)&:=[(\Phi_v^{-t})'(y)]^{1/2}f(\Phi_v^{-t}(y)),\label{eqn:ctransport}
\end{align}
where we take the right hand side as the definition for the notation $e^{tL_v}f$. The solution satisfies\footnote{This can be verified by using the identities $v(\Phi_v^{-t}(y))=(\Phi_v^{-t})'(y)v(y)$ (apply $\partial_s|_{s=0}$ to $\Phi_v^t(\Phi_v^s(y))=\Phi_v^{t+s}(y)$), along with its $\partial_y$ derivative, and $\partial_t(\Phi_v^{-t})'(y)=-v'(\Phi_v^{-t}(y))(\Phi_v^{-t})'(y)$ (apply $\partial_y$ to \eqref{eqn:char}).} $\partial_t(e^{tL_v}f)=L_v(e^{tL_v}f)$.
For $v=v_0$ or $v_1$, we may simplify some of the notation to write $\Phi_0^t:=\Phi_{v_0}^t$ and $\Phi_1^t:=\Phi_{v_1}^t$.
The right hand side is the standard formula for the pullback of $\Phi_v$ acting on half-densities \cite[\S9.1]{zworski2012semiclassical}.
For $v_{0}(y)=\sqrt{1-y^2}$ as in Lemma~\ref{lem:Dhv}, corresponding to the controlled $X$ rotation in the positive direction, we can solve $\partial_t\phi=v_0(\phi)$, $\phi(0)=y$, to obtain
\begin{align}\label{eqn:phi-w0}
\Phi_{0}^t(y)&=\sin(t+\arcsin y)=\sqrt{1-y^2}\sin t+y\cos t,
\end{align}
for $|t+\arcsin y|\le\pi/2$.
For the negative direction $X$ rotation, the solution is the same form but $t$ becomes $-t$.
For $v_1(y)=2y\sqrt{1-y^2}$ as in Lemma~\ref{lem:Dhv} (corresponding to TAT squeezing), using that $-\frac12\tanh^{-1}(\sqrt{1-\phi^2})$ is an antiderivative of $1/v_1(\phi)$ and that $\phi(0)=y$, we obtain
\begin{align}
\Phi_1^t(y)&=\frac{\sgn(y)}{\cosh(2t)\frac{1}{|y|}-\sinh(2t)\sqrt{\frac{1}{y^2}-1}}=\frac{2e^{2t}y(1+\sqrt{1-y^2})}{(e^{4t}-1)y^2+2(1+\sqrt{1-y^2})},
\end{align}
at least up to a branch point $y=\pm1$.
It will be convenient to set $y=\sin\theta$, and then use the stereographic projection coordinate (of the circle) $w:=\tan(\theta/2)$. Then
\begin{align}\label{eqn:w}
w(y):=\tan\left(\frac{\arcsin y}{2}\right)=\frac{y}{1+\sqrt{1-y^2}},\quad\text{and}\quad y(w)=\frac{2w}{1+w^2},
\end{align}
and
\begin{align}\label{eqn:phi-w}
\Phi_1^t(y)&=\frac{2e^{2t}w(y)}{1+e^{4t}w(y)^2},\quad w(\Phi_1^t(y))=e^{2t}w(y),
\end{align}
as long as $|e^{2t}w(y)|<1$ so that we avoid a branch point at $w$ or $y=\pm1$.
The last equality in \eqref{eqn:phi-w} is great, because it means that time evolution of a point $y$ by $\Phi_1^t$ is just linear multiplication by $e^{2t}$ when viewed in terms of the stereographic coordinate $w$.

\begin{proof}[Proof of Proposition~\ref{prop:supp}]
It is enough to prove the statement for $T_M^+$ only, since the negative version follows similarly by symmetry.
Using \eqref{eqn:ctransport} with the explicit solutions \eqref{eqn:phi-w0} and \eqref{eqn:phi-w}, it is straightforward to check how the support of $f_{M,g}$ changes under each step of $T_M^+$ to show \eqref{eqn:supp1}.
By \eqref{eqn:ctransport}, 
\begin{align}\label{eqn:suppL}
\supp(e^{tL_v}f)\subseteq\Phi_v^t(\supp(f)).
\end{align}
The initial packet $f_{M,g}$ is supported on $\{y:|y/\sigma_0|\le bg\}$. 
The support of $f_{M,g}$ thus changes as follows under the three steps of $T_M^+$.
\begin{enumerate}[1.]

\item First TAT squeezing: Let $y_0\in\supp(f_{M,g})$, and let $w_0:=w(y_0)$. Using \eqref{eqn:suppL}, we consider where $\Phi_{v_1}^{-s_1}$ sends $y_0$. Let $y_1:=\Phi_{v_1}^{-s_1}(y_0)$. By \eqref{eqn:phi-w}, this corresponds to the simple equation $w_1:=w(y_1)=e^{-2s_1}w_0$.

\item Controlled $X$ rotation: Starting from $w_1$, \eqref{eqn:suppL}, \eqref{eqn:phi-w0}, \eqref{eqn:w}, and the tangent angle-sum formula imply $w_1$ gets mapped to 
\begin{align*}
w_2&:=\tan\left(\frac{\arcsin(y_1)+\phi}{2}\right)=\frac{w_1+\tan(\phi/2)}{1-w_1\tan(\phi/2)}.
\end{align*}

\item Second TAT squeezing: Similar to the first TAT squeezing, we have $w_2$ sent to $w_3=w(y_3)=e^{2s_2}w_2$, with
\begin{align}\label{eqn:wy3}
w_3&=e^{2s_2}\frac{w_1+\tan(\phi/2)}{1-w_1\tan(\phi/2)},\quad y_3=\frac{2w_3}{1+w_3^2}.
\end{align}

\end{enumerate}
We just need to check that $1/3<y_3<1-\epsilon$ for sufficiently large $M$. We series expand the $w_i$s. Note that $w(y)=y/2+O(y^3)$, which is an adequate approximation when applied to $y_0\in\supp(f_{M,g})$ since it is near the $+Z$-pole, $y_0=O(\sigma_0bg)=o(1)$. Letting $u:=y_0/\sigma_0=y_0\sqrt{M}(1+o(1))$, the initial support of $f_{M,g}$ is $|u|\le bg$. Then using the explicit expressions for the parameters $s_1,\phi,s_2$, we get
\begin{align*}
w_1&=e^{-2s_1}w(y_0)=\frac{\kappa\log M}{M^{1/2}g}\frac{y_0}{2}(1+o(1))=\frac{\phi u}{2g}(1+o(1)),\\
w_3&=e^{2s_2}\frac{w_1+\tan(\phi/2)}{1-w_1\tan(\phi/2)}=\frac1\phi\frac{\frac{\phi u}{2g}(1+o(1))+\frac{\phi}{2}+O(\phi^3)}{1-\frac{\phi u}{2g}(1+o(1))O(\phi)}\\
&\hspace{3.6cm}=\left(\frac{u}{2g}+\frac{1}{2}\right)(1+o(1)),\numberthis\label{eqn:w-ug}
\end{align*}
using \eqref{eqn:phi-w} and provided that we have avoided the branch point (corresponding to $w=\pm1$) at the $Y$-poles. The $o(1)$ terms in \eqref{eqn:w-ug} depend on $g,M$, but are uniform over all $u$ in the support.
Since $|u|\le bg$, we have for $b=1/2$, $\frac14-o(1)\le w_3\le\frac34+o(1)$, so we avoid the branch point for sufficiently large $M$. Since $2w_3/(1+w_3^2)$ is increasing in $w_3$ for $0\le w_3\le1$, we obtain
\begin{align}
y_3=\frac{2w_3}{1+w_3^2}\ge \frac{\frac12-o(1)}{1+\frac{1}{16}+o(1)}=\frac{8}{17}+o(1),
\end{align}
which is $>1/3$ for sufficiently large $M$. Additionally, 
\begin{align}
y_3=\frac{2w_3}{1+w_3^2}\le\frac{\frac{3}{2}+o(1)}{1+\frac{9}{16}-o(1)}=\frac{24}{25}+o(1),
\end{align}
which is $\le1-\epsilon$ for sufficiently large $M$.
\end{proof}

\section{Remaining part of the proof of Theorem~\ref{thm:tat}: Discrete-continuum error bounds}\label{sec:dc-error}

In this section, we prove the final step of the proof of Theorem~\ref{thm:tat} listed in Section~\ref{sec:proof}, 
\begin{enumerate}
\item[(3)] Discrete-continuum error bounds.
\end{enumerate}
We need to estimate the difference between the discrete evolution and the continuum evolution of the Gaussian packet,
\begin{align*}
\|V_M^\pm S_hf_{M,g}-S_hT_M^\pm f_{M,g}\|=\|e^{s_2D_{h,v_1}}e^{\pm\phi D_{h,v_0}}e^{-s_1D_{h,v_1}}S_hf_{M,g}-S_he^{s_2L_{v_1}}e^{\pm\phi L_{v_0}}e^{-s_1L_{v_1}}f_{M,g}\|_{\ell^2},
\end{align*}
recalling $(S_hf)_m=\sqrt{h}f(hm)$ is the discretization of a function $f:I\to\C$ at scale $h=2/(M+1)$. We will prove
\begin{prop}[discrete-continuum bound]\label{prop:dcbounds}
Let $h=2/(M+1)$ and $g=\Theta(\sqrt{\log\log M})$, and let $D_{h,v}$ and $L_v$ be defined as in Lemma~\ref{lem:Dhv}. Then
\begin{align}\label{eqn:dc}
\|e^{s_2D_{h,v_1}}e^{\pm\phi D_{h,v_0}}e^{-s_1D_{h,v_1}}S_hf_{M,g}-S_he^{s_2L_{v_1}}e^{\pm\phi L_{v_0}}e^{-s_1L_{v_1}}f_{M,g}\|_{\ell^2}&\le O\left(\frac{g^3}{(\log M)^2}\right).
\end{align}
\end{prop}

The proof will primarily be through various derivative and ODE transport estimates.
For $D$ one of the discrete generators and $L$ its continuum counterpart,
we will make use of the standard ODE Duhamel bound
\begin{align}\label{eqn:duhamel}
\|e^{tD}S_hf-S_he^{tL}f\|_{\ell^2}&\le \int_0^{t}\|D S_hf_s-S_hLf_s\|_{\ell^2}\,ds,\quad f_s:=e^{sL}f,\quad t\ge0,
\end{align}
which reduces the estimate to studying the differences between the generators. 
The proof of \eqref{eqn:duhamel} is by integrating the $s$-derivative of $e^{(t-s)D}S_he^{sL}f$, and noting that $e^{(t-s)D}$ is unitary since $D$ is skew-adjoint.
We will rigorously bound the relevant terms $\|D S_hf_s-S_hLf_s\|_{\ell^2}$ in Section~\ref{subsec:dc-error-proof}, but we will first give heuristic estimates for each of the three steps in $V_M$. 

\subsection{Heuristic explanation}\label{subsec:heuristic}
For the heuristics, we work only in the small $y$ regime. 
In this regime, TAT squeezing acts as stretching or contracting in the $x$ or $y$ direction by a factor $e^{\pm 2s}$.
In the final TAT step, where we move the packets to macroscopic $y$, this small $y$ approximation doesn't hold. But since the packets increase in $y$-width in this step, the continuum approximation is expected to improve, and so the error should be small. We rigorously bound everything in the following Section~\ref{subsec:dc-error-proof}.

Heuristically, letting $y=y_m=hm$, Taylor expansion (cf.~\eqref{eqn:DL0}) shows the finite-difference operator $D_{h,v}$ \eqref{eqn:Dhv} and its continuum approximation $L_v$ \eqref{eqn:Lv} roughly satisfy
\begin{align*}
(D_{h,v}S_hf)_m&\approx \begin{multlined}[t]S_h\Big[-v(y)f'(y)-\frac12v'(y)f(y)+\frac{h}{2}v'(y)f'(y)+h^2O(vf'''+v''f')\\
-\frac{h}{2}v'(y)f'(y)+h^2O(v'f''+v'''f)\Big]
\end{multlined}\\
&=S_hL_vf(y)+h^2O(f'''v+f''v'+f'v''+fv''').\numberthis
\end{align*}
For $f$ a packet with current $y$-width $\sigma$, then we have the scaling $f=\sigma^{-1/2}F(\cdot/\sigma)$ for a fixed/non-scaling function $F$, and so the $r$th derivative of $f$ scales as $f^{[r]}\sim \sigma^{-1}O(f^{[r-1]})$. Then for $O(1)$ size derivatives of $v$, and using $\|f\|_{\ell^2}=1$, the above becomes
\begin{align}
D_{h,v}S_hf&\approx S_hL_vf + h^2\sigma^{-3}O(|vf|)\quad\Longrightarrow\quad
\|D_{h,v}S_hf-S_hL_vf\|_{\ell^2}\lesssim h^2\sigma^{-3}O(|v|),\label{eqn:heuristic-DL}
\end{align}
where $\sigma$ indicates the current packet $y$-width, and $O(|v|)$ informally means the maximum typical value of $|v(y)|$ on the packet support.
Then starting from the target quantity in Proposition~\ref{prop:dcbounds} and using the triangle inequality followed by \eqref{eqn:duhamel} and \eqref{eqn:heuristic-DL}, we can heuristically estimate the error in each step as follows.
\begin{itemize}
\item First TAT unitary: The TAT squeezing sends $y$ to approximately $e^{-2s}y$, so a packet of width $\sigma_0$ becomes a packet of width approximately $\sigma(s):=e^{-2s}\sigma_0$. We also have $|v_1|\lesssim 2y\lesssim O(\sigma(s))$, since the typical $y$-coordinate for a packet centered at $y=0$ of width $\sigma(s)$ is $O(\sigma(s))$.
Then \eqref{eqn:heuristic-DL} gives error at time $s$ approximately $\lesssim h^2\sigma(s)^{-2}=h^2\sigma_0^{-2}e^{4s}$, and integrating with \eqref{eqn:duhamel} gives
\begin{align}
E_1\lesssim\int_0^{s_1}\frac{h^2}{\sigma_0^2}e^{4s}\,ds=O\left(\frac{g^2}{\kappa^2(\log M)^2}\right).
\end{align}

\item Controlled $X$ rotation: Here $\sigma(t)$ is fixed at $\sigma_*:=\sigma(s_1)=e^{-2s_1}\sigma_0=\kappa\frac{\log M}{Mg}$. Note $|v_0|\le1$. 
Then
\begin{align}
E_2\lesssim\int_0^\phi \frac{h^2}{\sigma_*^3}\,ds\lesssim \frac{\kappa\log M}{M}\frac{1}{M^2}\frac{M^3g^3}{\kappa^3(\log M)^3}=O\left(\frac{g^3}{\kappa^2(\log M)^2}\right).
\end{align}

\item Second TAT unitary: For small $y$, the TAT squeezing sends $y$ to approximately $e^{2s}y$. After the controlled $X$ rotation, the center of the packets at the end of the rotation are located at $y\approx\pm\phi=\pm g\sigma_*$. Thus the packet starts with width $\sigma_*$ and center $y\approx\pm\phi$, so we have $\sigma(s)\approx e^{2s}\sigma_*$, evolved center $y\approx e^{2s}\phi=\sigma(s)g$, and $|v_1|\lesssim 2y\lesssim O(\sigma(s)g)$, giving (if we assume the error bounds for macroscopic $y$ are also small)
\begin{align}
E_3&\lesssim\int_0^{s_2}\frac{h^2}{\sigma(s)^2}O(g)\,ds\lesssim\frac{g}{M^2\sigma_*^2}\int_0^{s_2}e^{-4s}\,ds=O\left(\frac{g^3}{\kappa^2(\log M)^2}\right).
\end{align}
Note however that for macroscopic $y$ values, the TAT map is noticeably nonlinear, and so $\sigma(s)=e^{2s}\sigma_*$ is not exactly the packet width due to Jacobian or derivative factors of the map; however the packet width can be shown to still be $\Theta(\sigma(s))$.
\end{itemize}
Thus the largest discrete-continuum approximation error term, which contributes to the TAT packet separation error \eqref{eqn:ypoles}, is $O\left(\frac{g^3}{(\log M)^2}\right)$, which matches with Proposition~\ref{prop:dcbounds}. The remaining errors in the TAT packet separation protocol are those from Proposition~\ref{prop:gstate}, which are $O(g^4/M)+O(1/(\log M)^2)$, which is smaller. Then the total error in Theorem~\ref{thm:tat}, as seen in Section~\ref{subsec:tat-proof}, is $O\left(\frac{g^3}{(\log M)^2}\right)$, as desired.

\subsection{Rigorous proof}\label{subsec:dc-error-proof}

In this section, we rigorously prove Proposition~\ref{prop:dcbounds} on the error estimates for the continuum approximation in each of the three steps of the TAT packet separation protocol. 
This consists of some higher order derivative and ODE transport estimates. 
Essentially, we just need to control the 3rd (and 4th) order error terms in the Taylor expansion \eqref{eqn:heuristic-DL}, which involves bounding 3rd and 4th order derivative terms for the evolved packet.

To start, we can apply the triangle inequality to get
\begin{align*}\numberthis\label{eqn:duhamel3}
\|e^{s_2D_{h,v_1}}e^{\pm\phi D_{h,v_0}}&e^{-s_1D_{h,v_1}}S_hf_{M,g}-S_he^{s_2L_{v_1}}e^{\pm\phi L_{v_0}}e^{-s_1L_{v_1}}f_{M,g}\|_{\ell^2} \\
&\le \begin{multlined}[t]
\|e^{s_2D_{h,v_1}}e^{\pm\phi D_{h,v_0}}(e^{-s_1D_{h,v_1}}S_hf_{M,g})-e^{s_2D_{h,v_1}}e^{\pm\phi D_{h,v_0}}(S_he^{-s_1L_{v_1}}f_{M,g})\|_{\ell^2}\\
+\|e^{s_2D_{h,v_1}}(e^{\pm\phi D_{h,v_0}}S_he^{-s_1L_{v_1}}f_{M,g}) - e^{s_2D_{h,v_1}}(S_he^{\pm\phi L_{v_0}}e^{-s_1L_{v_1}}f_{M,g})\|_{\ell^2}\\
+\|e^{s_2D_{h,v_1}}S_he^{\pm\phi L_{v_0}}e^{-s_1L_{v_1}}f_{M,g}-S_he^{s_2L_{v_1}}e^{\pm\phi L_{v_0}}e^{-s_1L_{v_1}}f_{M,g})\|_{\ell^2}.
\end{multlined}
\end{align*}
Let
\begin{align}\label{eqn:f123}
f_s^{(1)}&:=e^{-sL_{v_1}}f_{M,g},\quad 
f_s^{(2)}:=e^{\pm sL_{v_0}}(e^{-s_1L_{v_1}}f_{M,g}),\quad\text{and}\quad
f_s^{(3)}:=e^{sL_{v_1}}(e^{\pm\phi L_{v_0}}e^{-s_1L_{v_1}}f_{M,g}).
\end{align}
Using that $e^{sD_{h,v}}$ is unitary and applying Duhamel's formula \eqref{eqn:duhamel}, we get that \eqref{eqn:duhamel3} is bounded above by
\begin{multline}\label{eqn:int3}
\int_0^{s_1}\|S_hL_{v_1}f_{s}^{(1)} - D_{h,v_1}S_h f_{s}^{(1)} \|_{\ell^2}\,ds
+\int_0^{\phi}\|D_{h,v_0}S_h f_{s}^{(2)}-S_hL_{v_0}f_{s}^{(2)}\|_{\ell^2}\,ds \\
+\int_0^{s_2}\| S_hL_{v_1}f_s^{(3)}-D_{h,v_1}S_hf_s^{(3)}\|_{\ell^2}\,ds.
\end{multline}

In order to estimate the norms in \eqref{eqn:int3}, we start with a lemma which gives useful formulas for $f_s^{(1)},f_s^{(2)},f_s^{(3)}$.
Recall the initial packet $f_{M,g}$ is given by $f_{M,g}(\sigma_0u)=\sigma_0^{-1/2}F_g(u)$, as defined in \eqref{eqn:fm} in Section~\ref{subsec:g0}.
\begin{lem}[packet evolution]\label{lem:Y}
Let $\Phi_{(j)}^s$, $j=1,2,3$, be the classical flows \eqref{eqn:char} corresponding to the 3 stages in the TAT-rotate-TAT protocol, with continuum generators denoted by $L^{(j)}$.
The time-evolved $y$-coordinate of an initial point $y_0=\sigma_0u$ under the classical flows is given by $Y_{1,s}(u)=\Phi_{(1)}^s(\sigma_0u)$, and for $j=2,3$, $Y_{j,s}(u)=\Phi_{(j)}^s(Y_{j-1,*}(u))$, where $*$ denotes the endpoint value of $s$.
For $f^{(1)}_s,f^{(2)}_s,f^{(3)}_s$ as in \eqref{eqn:f123}, 
\begin{align}\label{eqn:fY}
f_s^{(j)}(Y_{j,s}(u)) =[Y_{j,s}'(u)]^{-1/2}F_g(u),\quad \text{any $|u|\le bg$}.
\end{align}
\end{lem}
\begin{proof}
Note we can write $f^{(j)}_s=e^{sL^{(j)}}f^{(j)}_0$.
The case $j=1$ follows from \eqref{eqn:ctransport} and the change of variables $y=Y_{1,s}(u)$. Since $Y_{1,s}(u)=\Phi_{(1)}^s(\sigma_0u)$, then $(\Phi_{(1)}^{-s})'(Y_{1,s}(u))\cdot Y'_{1,s}(u)=\sigma_0$, and we obtain
\begin{align*}
f_s^{(1)}(Y_{1,s}(u))\equiv(e^{sL^{(1)}}f_{M,g})(Y_{1,s}(u))&=[(\Phi_{(1)}^{-s})'(Y_{1,s}(u))]^{1/2}\sigma_0^{-1/2}F_g(\Phi_{(1)}^{-s}(Y_{1,s}(u))/\sigma_0)\\
&=[Y'_{1,s}(u)]^{-1/2}F_g(u).\numberthis\label{eqn:Y1}
\end{align*}
For $j=2$, we have $Y_{2,t}(u)=\Phi_{(2)}^t(Y_{1,s_1}(u))$, and so $(\Phi_{(2)}^{-t})'(Y_{2,t}(u))\cdot Y_{2,t}'(u)=Y_{1,s_1}'(u)$. This gives
\begin{align*}
f_t^{(2)}(Y_{2,t}(u))\equiv(e^{tL^{(2)}}f^{(1)}_{s_1})(Y_{2,t}(u))&=[(\Phi_{(2)}^{-t})'(Y_{2,t}(u))]^{1/2}f^{(1)}_{s_1}(\Phi_{(2)}^{-t}(Y_{2,t}(u)))\\
&= [Y_{2,t}'(u)]^{-1/2}F_g(u),\numberthis
\end{align*}
where in the last equality we also used \eqref{eqn:Y1} to evaluate $f^{(1)}_{s_1}$ at $\Phi_{(2)}^{-t}(Y_{2,t}(u))=Y_{1,s_1}(u)$.
The case $j=3$ follows in the same way.
\end{proof}

Next, the following Lemma~\ref{lem:DL-approx}(ii) essentially gives the rigorous version of the heuristic estimate \eqref{eqn:heuristic-DL}.
\begin{lem}[$D_{h,v},L_v$ comparison]\label{lem:DL-approx}
Let $f:[-1,1]\to\C$ be $C^4$. Suppose $\supp f$ is contained in a compact subset of $(-1,1)$, and that $h$ is small enough so the $h$-neighborhood of $\supp f$ is also contained in a compact subset $K\subset(-1,1)$.
\begin{enumerate}[(i)]
\item For $v$ $C^4$ in a neighborhood of the $h$-neighborhood of $\supp f$, we have the general bound
\begin{align}\label{eqn:DLbound}
\|D_{h,v}S_hf-S_hL_vf\|_{\ell^2}&\le Ch^2(P_{3}(v,f;h) + hP_4(v,f;h)),
\end{align}
where $P_3$ and $P_4$ are third and fourth order derivative norms defined via
\begin{align}\label{eqn:Pk}
P_k(v,f;h)&:=\sup_{|t|\le h}\sum_{\substack{a+b=k\\a,b\ge0}}\|v^{[a]}(\cdot+t/2)f^{[b]}(\cdot+t)\|_2,
\end{align}
where the notation $f^{[b]}$ denotes the $b$th derivative of a function $f$. 

\item  Suppose $f(Y(u))=[Y'(u)]^{-1/2}F(u)$ for some $F\in C^4(I)$ with $I\subset\R$ a compact interval, $Y:I\to(-1,1)$ an increasing $C^5$ function, and $f=0$ outside of $Y(I)$. Suppose that
\begin{align}\label{eqn:F4r}
\sum_{r=0}^4\|F^{[r]}\|_2=O(1),\quad\sum_{r=0}^4\|uF^{[r]}\|_2=O(1),
\end{align}
where $F^{[r]}$ denotes the $r$th derivative of $F$, and also that for some $h\le\sigma\le1$ and $\mu\ge0$, over $u\in I$, 
\begin{align}\label{eqn:Yd}
Y'=\Theta(\sigma),\quad |Y^{[r]}(u)|=O(\sigma)\text{ for $r=2,3,4,5$},\quad\text{and } |Y(u)|\le C\sigma(|u|+\mu).
\end{align}
Then for $v_0(y)=\sqrt{1-y^2}$ and $v_1(y)=2y\sqrt{1-y^2}$ as in \eqref{eqn:Dv}, we have
\begin{align}\label{eqn:dl-scaling}
\begin{aligned}
\|D_{h,v_0}S_hf-S_hL_{v_0}f\|_{\ell^2}&\le C_1\sigma^{-3}h^2,\\
\|D_{h,v_1}S_hf-S_hL_{v_1}f\|_{\ell^2}&\le C_2\max(1,\mu)\sigma^{-2}h^2,
\end{aligned}
\end{align}
with the constants $C_1,C_2$ depending only on the constants in \eqref{eqn:F4r} and \eqref{eqn:Yd}, and on the compact $y$-interval $K$.
\end{enumerate}
\end{lem}

\begin{proof}
(i) Let 
\begin{align*}
(A_{h,v}f)(y):=\frac{v(y-h/2)f(y-h)-v(y+h/2)f(y+h)}{2h},
\end{align*}
so that $D_{h,v}S_hf=S_hA_{h,v}f$. 
In terms of $H_y(t):=v(y+t/2)f(y+t)$, we have
\begin{align}
A_{h,v}f(y)=\frac{H_y(-h)-H_y(h)}{2h},\quad\text{and}\quad L_vf(y)=-H_y'(0).
\end{align}
Taylor's integral remainder theorem applied to $H_y(\pm h)$ with 3rd order remainder gives
\begin{align}\label{eqn:taylor-remainder}
A_{h,v}f(y)-L_vf(y)&=-\frac{h^2}{4}\int_0^1(1-t)^2[H_y'''(ht)+H_y'''(-ht)]\,dt,
\end{align}
and so we obtain the $L^2$ estimate
\begin{align}\label{eqn:AL}
\|A_{h,v}f-L_vf\|_2&\le Ch^2P_3(v,f;h).
\end{align}

To turn this into an $\ell^2$ estimate with $S_h$, we note that for $r\in C^1(I)$ and any grid of spacing $h$ on an interval $I=[a,a+hn]$, 
\begin{align}\label{eqn:rsampling}
h\sum_{m=1}^n|r(z_m)|^2&\le 2\|r\|_{L^2}^2+2h^2\|r'\|_{L^2}^2,
\end{align}
where $z_m:=a+hm$. This follows by bounding each $|r(z_m)|^2$ using the fundamental theorem of calculus and Cauchy--Schwarz, as
\begin{align*}
|r(z_m)|^2=\left|r(x)+\int_{x}^{z_m}r'(t)\,dt\right|^2
&\le 2|r(x)|^2+2\left|\int_{z_{m-1}}^{z_m}r'(t)\,dt\right|^2\\
&\le 2|r(x)|^2+2h\int_{z_{m-1}}^{z_m}|r'(t)|^2\,dt,\quad\text{ any $x\in[z_{m-1},z_m]$}.\numberthis
\end{align*}
Integrating over $x\in[z_{m-1},z_m]$ and summing over $m$ gives \eqref{eqn:rsampling}.

Thus
\begin{align}
\|S_hA_{h,v}f-S_hL_vf\|_{\ell^2}&\le \sqrt{2}\|A_{h,v}f-L_vf\|_2+\sqrt{2}h\|\partial_y(A_{h,v}f-L_vf)\|_2.
\end{align}
Differentiating \eqref{eqn:taylor-remainder} with respect to $y$ then gives
\begin{align}
\|\partial_y(A_{h,v}f-L_vf)\|_2&\le Ch^2P_4(v,f;h).
\end{align}
With \eqref{eqn:AL}, this implies \eqref{eqn:DLbound}.

\vspace{2mm}
(ii) To prove \eqref{eqn:dl-scaling}, we apply \eqref{eqn:DLbound}, so we need to estimate $P_3$ and $P_4$. 
We first estimate the derivatives of $f$. Let $Y'(u)=:\sigma q(u)$, so $q=\Theta(1)$ and $q^{[r]}=O(1)$ for $r=1,2,3,4$. This also implies $(q^{-1})^{[r]}=O(1)$ for $r\le4$, for example by successive chain rules. Writing $f(Y(u))=\sigma^{-1/2}q(u)^{-1/2}F(u)$ and differentiating with respect to $u$, we get
\begin{align*}
f'(Y(u))=\sigma^{-1}q(u)^{-1}\sigma^{-1/2}\partial_u(q(u)^{-1/2}F(u)).
\end{align*}
Continuing to differentiate, we see
\begin{align}
f^{[r]}(Y(u))&=\sigma^{-r-1/2}\sum_{j=0}^rb_{j,r}(u)F^{[j]}(u),
\end{align}
where $b_{j,r}(u)$ is a polynomial in $q,q^{-1/2}$, and derivatives of $q$.
Using \eqref{eqn:F4r} and \eqref{eqn:Yd}, the bounds on $q$ and its derivatives, and the change of variables $y=Y(u)$, $dy=\sigma q(u)\,du$, then gives for $r=1,2,3,4$,
\begin{align}\label{eqn:fderiv}
\|f^{[r]}\|_2&=O(\sigma^{-r}), \quad \|yf^{[r]}\|_2=O(\sigma^{1-r}(1+\mu)).
\end{align}

We next look for pointwise bounds on $v_0^{[a]},v_1^{[a]}$. Since we work on a compact interval $K$ in $(-1,1)$, the first four derivatives of $v_0,v_1$ are bounded on $K$. For $v_0$ we thus obtain
\begin{align}
P_3(v_0,f;h)&\le O(\sigma^{-3}),\quad P_4(v_0,f;h)\le O(\sigma^{-4}),
\end{align}
which gives the first equation of \eqref{eqn:dl-scaling} using \eqref{eqn:DLbound} and that $h\sigma^{-1}\le1$. This will be good enough for the controlled $X$ rotation estimate.

For $v_1$, we will need to do a better estimate. Since the Taylor expansion of $v_1$ at zero has only odd power terms, or by direct differentiation, the even derivatives satisfy $|v_1^{[2]}(y)|\le C_K|y|$ and $|v_1^{[4]}(y)|\le C_K|y|$ on the compact $y$-interval $K\subset(-1,1)$ which contains $Y(I)$ and its $h/2$ neighborhood for sufficiently small $h$. Using this pointwise bound for $v^{[a]}$ in \eqref{eqn:Pk} when $a$ is even ($a\in\{0,2,4\}$) gives for $|t|\le h$,
\begin{align*}
\|v^{[a]}(\cdot+t/2)f^{[b]}(\cdot+t)\|_2=\|v^{[a]}(\cdot-t/2)f^{[b]}(\cdot)\|_2
&\le C_K\|(|y|+h)f^{[b]}(y)\|_2\\
&=O(\sigma^{1-b}(1+\mu)+h\sigma^{-b})=O((1+\mu)\sigma^{-b+1}),\numberthis
\end{align*}
since $h\le\sigma$.
For $a$ odd, we just use a constant pointwise upper bound on $|v^{[a]}(y)|$, which gives the bound $\|v^{[a]}(\cdot+t/2)f^{[b]}(\cdot+t)\|_2\le O(\sigma^{-b})$. 
Thus for $k=3,4$ and considering $a+b=k$, we get
\begin{align*}
P_k(v_1,f;h)&\le O(\sigma^{-k+1})+O(\mu\sigma^{-k+1})=O(\max(\mu,1)\sigma^{-k+1}),\numberthis
\end{align*}
which using \eqref{eqn:DLbound} and $h\sigma^{-1}\le1$ gives the second equation of \eqref{eqn:dl-scaling}.
\end{proof}

We next want to apply Lemma~\ref{lem:DL-approx}(ii) to the functions in \eqref{eqn:int3}. This is essentially the rigorous version of the estimates for each protocol step following \eqref{eqn:heuristic-DL}. Here we verify the conditions in Lemma~\ref{lem:DL-approx}(ii) hold with the appropriate scales $\sigma=\sigma(s)$. Note that due to derivative/Jacobian terms in the TAT flow, this $\sigma$ may not be the precise packet width when we move away from the initial pole; however, we will see the packet width is still $\Theta(\sigma)$.

\begin{prop}[TAT protocol $D_{h,v},L_v$ bounds]\label{prop:tat-dl}
Consider the setting of Lemma~\ref{lem:Y} and Lemma~\ref{lem:DL-approx}(ii). For the three steps of the TAT packet separation protocol and functions $f_s^{(1)}$, $f_s^{(2)}$, and $f_s^{(3)}$ in \eqref{eqn:f123}, we have:
\begin{align}\label{eqn:array}
\begin{array}{lllll}
\text{step} & f & Y  & \sigma &\mu \\\hline
\text{first TAT} & f_s^{(1)} & Y_{1,s} &\sigma(s)=e^{-2s}\sigma_0 & O(1)\\
\text{controlled rotation} &f_s^{(2)} & Y_{2,s} & \sigma(s)=\sigma_* & O(g) \\
\text{second TAT}& f_s^{(3)} & Y_{3,s} &\sigma(s)=e^{2s}\sigma_*& O(g)
\end{array}.
\end{align}
Thus
\begin{align}\label{eqn:normbound3}
\begin{aligned}
\|D_{h,v_1}S_hf_s^{(1)}-S_hL_{v_1}f_s^{(1)}\|_{\ell^2}&\le O(\sigma_0^{-2}h^2)e^{4s},\\
\|D_{h,v_0}S_hf_s^{(2)}-S_hL_{v_0}f_s^{(2)}\|_{\ell^2} &\le O(\sigma_*^{-3}h^2), \\
\|D_{h,v_1}S_hf_s^{(3)}-S_hL_{v_1}f_s^{(3)}\|_{\ell^2} &\le O(g\sigma_*^{-2}h^2)e^{-4s}.
\end{aligned}
\end{align}
\end{prop}
\begin{proof}
The proof involves using Lemma~\ref{lem:Y} and Lemma~\ref{lem:DL-approx}(ii) and estimating the derivatives of $Y_{j,s}$ to fill in the table \eqref{eqn:array}. 
Recall the classical flows $\Phi_0^s,\Phi_1^s$ from \eqref{eqn:phi-w0} and \eqref{eqn:phi-w}.
For TAT squeezing, recall also $w$ from \eqref{eqn:w}, and its inverse which we now denote by $J$,
\begin{align}
w(y)=\frac{y}{1+\sqrt{1-y^2}},\qquad J(w):=\frac{2w}{1+w^2},\quad J'(w)=\frac{2(1-w^2)}{(1+w^2)^2}.
\end{align}
Recall in particular that $w(\Phi_1^s(y))=e^{2s}w(y)$, i.e. $\Phi_1^s(y)=J(e^{2s}w(y))$, and also that $w(y)=y/2+O(y^3)$.
\begin{enumerate}[1.,leftmargin=*]
\item We start with the first TAT squeezing step. 
For $0\le s\le s_1$, as in Lemma~\ref{lem:Y} we have
\begin{align}
Y_{1,s}(u)&=\Phi_1^{-s}(\sigma_0u)=J(e^{-2s}w(\sigma_0u)),
\end{align}
and $f_s^{(1)}(Y_{1,s}(u))=[Y_{1,s}'(u)]^{-1/2}F_g(u)$ by Lemma~\ref{lem:Y}. 

To be able to apply Lemma~\ref{lem:DL-approx}(ii), we bound the derivatives of $F_g(u)$. Recall the initial packet is $f_{M,g}$, where $f_{M,g}(\sigma_0u)=\sigma_0^{-1/2}F_g(u)$ for $F_g(u)=\chi(u/g)G(u)$, $G(u)=\frac{e^{-u^2/4}}{(2\pi)^{1/4}}$, as in Section~\ref{subsec:g0}. Taking $r$ derivatives of $F_g$, every term in the resulting sum is of the form $g^{-j}\chi^{[j]}(u/g) G^{[r-j]}(u)$. Since $\chi$ is smooth with bounded derivatives, $G$ is fixed, and $g^{-j}\le1$ for large enough $M$, we have the starting bounds
\begin{align}\label{eqn:F4-start}
\sum_{r=0}^4\|F_g^{[r]}\|_2=O(1),\quad \sum_{r=0}^4\|uF_g^{[r]}\|_2=O(1),
\end{align}
which fulfills requirement \eqref{eqn:F4r}.

Now we bound the derivatives of $Y_{1,s}$.
Let $W_{1,s}(u):=e^{-2s}w(\sigma_0u)$, so $Y_{1,s}(u)=J(W_{1,s}(u))$. 
For $I=\{u:|u|\le bg\}$, we have $\sigma_0u=o(1)$. For $y=\sigma_0u=o(1)$, then $w'(y)=\frac12+O(y^2)$, and the higher order derivatives of $w$ are also bounded since $\sigma_0u$ stays away from $\pm1$.
Letting\footnote{Note that the scale $\sigma_1(s)=e^{-2s}\sigma_0$ agrees with the heuristics: For small $y$, we have $w(y)=y/2+O(y^3)$. Since the TAT squeezing sends $w\mapsto e^{-2s}w$, it then also sends $y$ to approximately $e^{-2s}y$ (cf. \eqref{eqn:phi-w}). So an initial packet of $y$-width $\sigma_0$ becomes a packet of $y$-width approximately $\sigma_1(s)=e^{-2s}\sigma_0$.} $\sigma_1(s):=e^{-2s}\sigma_0$, we can then estimate
\begin{align*}
W_{1,s}'(u)&=e^{-2s}\sigma_0(1/2+o(1))=\Theta(\sigma_1(s)),\\
W_{1,s}^{[r]}(u)&=e^{-2s}\sigma_0^rw^{[r]}(\sigma_0u)=O(\sigma_1(s)),\numberthis\\
W_{1,s}(u)&= e^{-2s}\left(\frac12(\sigma_0u)+O(\sigma_0^3u^3)\right)=O(\sigma_1(s)u)=o(1).
\end{align*}
Since $W_{1,s}(u)=o(1)$ stays away from $\pm1$, by taking derivatives of $J$, which are bounded, and using the chain rule for higher order derivatives (e.g. Fa\`a di Bruno's formula), we obtain the same estimates for $Y_{1,s}$ as well: 
\begin{align}\label{eqn:Yderiv}
Y_{1,s}'(u)=\Theta(\sigma_1(s)),\quad Y_{1,s}^{[r]}(u)=O(\sigma_1(s)),\quad\text{for $r=2,3,4,5$},
\end{align}
noting $J'$ is bounded away from $0$ for input bounded away from $\pm1$.
Additionally, since $|W_{1,s}(u)|\le C\sigma_1(s)u=o(1)$ and $|J(w)|\le2|w|$, we see $|Y_{1,s}(u)|\le C\sigma_1(s)|u|$ for $|u|\le bg$.
Taking $\sigma=\sigma_1(s)$ and $\mu=0$ gives the first row of the table \eqref{eqn:array}.

\item For the controlled $X$ rotation, it is enough to consider the positive $X$ rotation, since the negative $X$ rotation will have the same estimates. 
For the $X$ rotation, the scale $\sigma_*=\sigma_1(s_1)=\kappa\frac{\log M}{Mg}$ will be constant.
Letting $Y_*:=Y_{1,s_1}$, we have from \eqref{eqn:phi-w0},
\begin{align}
Y_{2,t}(u)&=\Phi_0^t(Y_{1,s_1}(u))=Y_*(u)\cos t+\sqrt{1-Y_*(u)^2}\sin t.
\end{align}
Let $R_t(y):=y\cos t+\sqrt{1-y^2}\sin t$, so that $Y_{2,t}(u)=R_t(Y_*(u))$ and $Y_{2,t}'(u)=R_t'(Y_*(u))Y_*'(u)$. 
From the end of the previous step with $s=s_1$, we see $|Y_*(u)|\le C\sigma_*|u|\le C\kappa\frac{\log M}{M}=C\phi$ for $|u|\le bg$.
For $|y|\le C\phi$ and $|t|\le \phi$, we have
\begin{align}
\begin{aligned}
R_t'(y)&=\cos t-\frac{y}{\sqrt{1-y^2}}\sin t=1+O(\phi^2)=\Theta(1),\\
R_t^{[r]}(y)&=O(1)\sin t=O(\phi),\quad r=2,3,4,5,
\end{aligned}
\end{align}
again using that $y$ stays away from $\pm1$ to obtain a uniform bound on the higher order derivatives of $\sqrt{1-y^2}$.
Then using $Y_{2,t}'(u)=R_t'(Y_*(u))Y_*'(u)$ and the chain rule for higher order derivatives, along with \eqref{eqn:Yderiv} with $s=s_1$ to estimate the derivatives of $Y_*=Y_{1,s}$, we obtain
\begin{align}\label{eqn:Yderiv2}
Y_{2,t}'=\Theta(\sigma_*),\quad Y_{2,t}^{[r]}(u)=O(\sigma_*),\;r=2,3,4,5.
\end{align}
Additionally,
\begin{align*}
|Y_{2,t}(u)|&\le |Y_*(u)|+|\sin t|\\
&\le C\sigma_*|u|+\phi \le C\sigma_*(|u|+g).\numberthis\label{eqn:Y2mu}
\end{align*}
This completes the second row of the table \eqref{eqn:array}.

\item For the final TAT squeezing, we consider the positive $X$ rotation from step 2. Let $Y_{**}:=Y_{2,\phi}$, so that 
\begin{align}
Y_{3,r}(u)=\Phi_1^r(Y_{**}(u))&=J(e^{2r}w(Y_{**}(u))).
\end{align}
Let $W_{3,r}(u):=e^{2r}w(Y_{**}(u))$ and $\sigma_3(r):=e^{2r}\sigma_*$.
From the proof of Proposition~\ref{prop:supp}, we know for $|u|\le bg$, the ending $w$-coordinate (and all earlier $w$-coordinates since $e^{2r}w$ is increasing in $r$) satisfies
\begin{align}
|W_{3,s_2}(u)|=|e^{2s_2}w(Y_{**}(u))|\le\frac34+o(1),
\end{align}
in particular it is bounded away from $\pm1$.
Using that $|Y_{**}(u)|\le C\sigma_*(|u|+g)\le C\phi=o(1)$, along with \eqref{eqn:Yderiv2}, we can estimate
\begin{align}
W_{3,r}'(u)&=e^{2r}w'(Y_{**}(u))Y_{**}'(u)= e^{2r}\left(\frac12+O(\phi^2)\right)\Theta(\sigma_*)=\Theta(\sigma_3(r)),
\end{align}
along with the higher-order derivative estimates
\begin{align}
W_{3,r}^{[k]}(u)&=e^{2r}O(\sigma_*)=O(\sigma_3(r)),
\end{align}
since the derivatives $w^{[k]}(y)$ are bounded for $y$ bounded away from $\pm1$.

Also since $W_{3,r}(u)$ stays bounded away from $\pm1$, considering derivatives of $J$ gives the same derivative estimates for $Y_{3,r}(u)=J(W_{3,r}(u))$. 
As $|J(w)|\le 2|w|$ and $|w(y)|\le |y|$, we also obtain, with \eqref{eqn:Y2mu},
\begin{align}
|Y_{3,r}(u)|\le 2|W_{3,r}(u)|= 2e^{2r}|w(Y_{**}(u))|\le 2e^{2r}|Y_{**}(u)|\le C\sigma_3(r)(|u|+g).
\end{align}
This completes the last row of the table \eqref{eqn:array}.

\end{enumerate}
Finally, we can check that $Y_{j,s}'(u)>0$ so each $Y_{j,s}$ is increasing. 
Also, recalling $h=2/(M+1)$, we see $h\le \sigma_j(s)\le 1$ for sufficiently large $M$, the smallest $\sigma_j(s)$ being $\sigma_2(s)=\sigma_*=\kappa\frac{\log M}{Mg}$.
Thus \eqref{eqn:normbound3} follows from Lemma~\ref{lem:DL-approx}(ii), noting that all three steps can use the same fixed compact $y$-interval $K\subset(-1,1)$, for sufficiently large $M$.
\end{proof}

\begin{proof}[Proof of Proposition~\ref{prop:dcbounds}]
Recall the parameter choices of $s_1,\phi,s_2$ from \eqref{eqn:svals}, and that $\sigma_*=\sigma_0e^{-2s_1}=\frac{\kappa\log M}{Mg}$.
Putting the bounds \eqref{eqn:normbound3} from Proposition~\ref{prop:tat-dl} into \eqref{eqn:int3} shows that 
\begin{align*}
\|e^{s_2D_{h,v_1}}e^{\pm\phi D_{h,v_0}}&e^{-s_1D_{h,v_1}}S_hf_{M,g}-S_he^{s_2L_{v_1}}e^{\pm\phi L_{v_0}}e^{-s_1L_{v_1}}f_{M,g}\|_{\ell^2} \\
&\le O(\sigma_0^{-2}h^2)\int_0^{s_1}e^{4s}\,ds +O(\sigma_*^{-3}h^2)\int_0^\phi ds + O(g\sigma_*^{-2}h^2)\int_0^{s_2}e^{-4s} \,ds\\
&\le O\left(\frac{g^2}{\kappa^2(\log M)^2}\right)+O\left(\frac{g^3}{\kappa^2(\log M)^2}\right)=O\left(\frac{g^3}{\kappa^2(\log M)^2}\right),\numberthis
\end{align*}
as in the heuristic argument in Section~\ref{subsec:heuristic}.
\end{proof}

\subsection{Finishing the proof of Theorem~\ref{thm:tat}}\label{subsec:tat-proof}
We can finally finish
\begin{proof}[Proof of Theorem~\ref{thm:tat}]
By Proposition~\ref{prop:supp}, we know that for sufficiently large $M$, the continuum evolution $T_M^\pm f_{M,g}$ is supported entirely in $\{y>1/3\}$ or $\{y<-1/3\}$.
Thus, recalling we work in the $Y$-Dicke basis, we have $\|(1-\Pi_{\pm,1/3})S_hT_M^\pm f_{M,g}\|_{\ell^2}=0$.

Applying Proposition~\ref{prop:gstate} on Gaussian state approximation, followed by the discrete-continuum bound Proposition~\ref{prop:dcbounds}, and finally the above observation, we obtain, for $g=\sqrt{8\log\log M}/a$,
\begin{align*}
\|(1-\Pi_{\pm,1/3})V_M^\pm|0^M\rangle_A\|_{\ell^2}&\le \|(1-\Pi_{\pm,1/3})V_M^\pm(S_hf_{M,g})\|_{\ell^2}+O\left(\frac{g^4}{M}\right)+O(e^{-a^2g^2/4})\\ 
&\le \|(1-\Pi_{\pm,1/3})S_hT_M^\pm f_{M,g}\|_{\ell^2}+O\left(\frac{g^3}{(\log M)^2}\right)+O\left(\frac{g^4}{M}\right)+O\left(\frac{1}{(\log M)^2}\right)\\
&=O\left(\frac{g^3}{(\log M)^2}\right),\numberthis
\end{align*}
which gives \eqref{eqn:ypoles}.
\end{proof}

\section{Proof of Theorem~\ref{thm:qsp} (QSP sgn approximation)}\label{sec:qsp-proof}

In this section, we prove Theorem~\ref{thm:qsp} on the QSP approximation involving the $\sgn$ function.

\begin{proof} 
Let $\eta=\Theta(\gamma)$ to be chosen later.
From \cite[Corollary 6, discussion before \S B]{low2017hamiltonian} or \cite[Lemma~25 of arXiv v1]{gilyen2019quantum} (see also \cite{eremenko2007uniform} for non-constructive existence), one can construct an odd polynomial $p$ of (odd) degree $d'=O(\frac{1}{\eta}\log(1/\varepsilon))$ such that $|p(x)|\le1$ on $[-1,1]$, and
\begin{align}\label{eqn:psign}
\max_{x\in[-1,-\eta]\cup[\eta,1]}|p(x)-\sgn(x)|\le \varepsilon.
\end{align}
We want to build this approximation $p$ into a unitary matrix $Q_\delta(x)$ which is close to $I_2$ for $x\in[\eta,1]$, and close to $iX$ for $x\in[-1,-\eta]$. 
We use the constructions in \cite{haah2019product}, which involve Laurent polynomials $F(z)$ which are $\SU(2)$-valued for $z=e^{i\theta}$ on the unit circle. 
Since any $\SU(2)$ matrix can be written uniquely in the form $aI+ibX+icY+idZ$ for $a,b,c,d\in\R$ with $a^2+b^2+c^2+d^2=1$, we will be interested in constructing Laurent polynomials $F(z)=a(z)I+ib(z)X+ic(z)Y+id(z)Z$, with $a(e^{i\theta})\approx \oneb_{(0,\pi)}(\theta)$ (corresponding to $x\in[\eta,1]$) and $b(e^{i\theta})\approx \oneb_{(-\pi,0)}(\theta)$ (corresponding to $x\in[-1,-\eta]$).

To do this, we will use the polynomial $p$ with \cite[Lemma 4]{haah2019product}, which, given Laurent polynomials $a,b$ satisfying certain conditions, provides existence of Laurent polynomials $c,d$ of at most the same degree so that $F(z)\in\SU(2)$.
To satisfy the conditions of \cite[Lemma 4]{haah2019product}, we first replace $p$ with $(1-\varepsilon)p$, so it is bounded strictly below 1. This adds at most another $\varepsilon$ to the error in \eqref{eqn:psign}. We can define a Laurent polynomial $q$ via $q(e^{i\theta})=p(\sin\theta)$ on the unit circle, which extends to $\C\setminus\{0\}$, as well as Laurent polynomials $a(z)$ and $b(z)$,  via
\begin{align}
q(z):=p\left(\frac{z-z^{-1}}{2i}\right),\quad a(z):=\frac12(1+q(z)),\quad b(z):=\frac12(1-q(z)).
\end{align}
We see $a^2+b^2$ is reciprocal ($f(z)=f(1/z)$) and $a^2+b^2\le 1-\varepsilon(1-\varepsilon/2)<1$ on the unit circle, so we can apply \cite[Lemma 4]{haah2019product}. This gives Laurent polynomials $c,d$ with degree at most $d'$ such that
\begin{align}
F(z)=a(z)I+ib(z)X+ic(z)Y+id(z)Z.
\end{align}
Moreover, 
\begin{align}
\begin{aligned}
\|F(e^{i\theta})-I\|^2&=(a(e^{i\theta})-1)^2+b(e^{i\theta})^2+c(e^{i\theta})^2+d(e^{i\theta})^2=1-q(e^{i\theta}),\\
\|F(e^{i\theta})-iX\|^2&=a(e^{i\theta})^2+(b(e^{i\theta})-1)^2+c(e^{i\theta})^2+d(e^{i\theta})^2=1+q(e^{i\theta}),
\end{aligned}
\end{align}
so in operator norm,
\begin{align}\label{eqn:F-approx}
\begin{aligned}
F(e^{i\theta})&=I+O(\sqrt{\varepsilon}),\text{ for }\sin\theta\in[\eta,1]\\ 
F(e^{i\theta})&=iX+O(\sqrt{\varepsilon}), \text{ for }\sin\theta\in[-1,-\eta].
\end{aligned}
\end{align}
Take $\varepsilon=c\delta^2$ so the error terms are $\le \delta$. With this choice, the degree of $F$ is $d'=O(\frac1\eta\log(1/\delta))$.

It remains to show that $F$ has the desired product decomposition in \eqref{eqn:Qdelta}.
For this, we apply \cite[Theorem 2]{haah2019product} to the Laurent polynomial $\tilde F(z):=F(z^2)$, which has degree $n:=\deg\tilde F\le 2d'$. Then $\tilde F$ is in Haah's class $E_{n}$ of degree $n$ Laurent polynomials $L(z)=\sum_{j=-n}^{j=n}C_jz^j$ for $C_j$ $2\times2$ complex matrices, which also satisfy: (1) $L(e^{i\theta})\in\SU(2)$ on the unit circle, and (2) $L$ only has terms $z^j$ for $j\in\{-n,-n+2,\ldots,n-2,n\}$.  Then \cite[Theorem 2]{haah2019product} gives the decomposition
\begin{align}
\tilde F(z)=F(z^2)=U_0\prod_{j=1}^{n}E_{P_j}(z),\quad\text{where } E_P(z):=zP+z^{-1}(I-P),
\end{align}
for some rank-1 projections $P_j$. For $z=e^{i\theta/2}$, we have $E_P(z)=e^{i\theta(2P-I)/2}$ by considering the action on the kernel and range of the projection $P$. Then letting $\theta=\pi x/2$, recalling \eqref{eqn:F-approx}, taking  $d=n\le 2d'$, $\eta=\sin(\pi\gamma/2)$, and
\begin{align}
Q_\delta(x):= F(e^{i\pi x/2}),
\end{align}
gives the theorem.
\end{proof}

\vspace{2mm}
\noindent
\textbf{Acknowledgments.}
This project used GPT-5.6 Sol Pro to come up with the protocol; GPT-5.6 Sol was also used for coming up with proof ideas and for general checking and proofreading. The paper was written by the authors and all results and proofs were checked and validated by the authors, who are fully responsible for the final content. L.S., T.U., T.C.M.\ and A.V.G.\ were supported in part by the DoE ASCR Quantum Testbed Pathfinder program (award No.~DE-SC0024220), NSF STAQ program, AFOSR MURI, NQVL:QSTD:Design:FTL, NSF QLCI (award No.~OMA-2120757), ARL (W911NF-24-2-0107), and ONR MURI. L.S., T.U., T.C.M.\ and A.V.G.\ also acknowledge support from the U.S.~Department of Energy, Office of Science, Accelerated Research in Quantum Computing, Fundamental Algorithmic Research toward Quantum Utility (FAR-Qu) and the U.S.~Department of Energy, Office of Science, National Quantum Information Science Research Centers, Quantum Systems Accelerator (award No.~DE-SCL0000121).

\bibliographystyle{amsalpha_edit}
\bibliography{ghz.bib}

\end{document}